\documentclass[manuscript,screen,nonacm]{acmart}
\AtBeginDocument{%
  }

\setcopyright{none}
\renewcommand\footnotetextcopyrightpermission[1]{}

\usepackage{thmtools}
\usepackage{thm-restate}
\usepackage[ruled]{algorithm2e}
\usepackage{cleveref}
\usepackage{xspace}
\usepackage{tikz}
\usepackage{amsmath}
\allowdisplaybreaks[4]
\usepackage{subcaption}

\begin{document}

%%
%% The "title" command has an optional parameter,
%% allowing the author to define a "short title" to be used in page headers.
\title{Fractional Fully Online Matching}
\titlenote{Preliminary versions of parts of this article appeared in SODA 2019, FOCS 2020, and EC 2024~\cite{soda/HuangPTTWZ19,focs/HuangTWZ20,DBLP:conf/sigecom/Tang024}.}

%%
%% The "author" command and its associated commands are used to define
%% the authors and their affiliations.
%% Of note is the shared affiliation of the first two authors, and the
%% "authornote" and "authornotemark" commands
%% used to denote shared contribution to the research.
\author{Zhiyi Huang}
\affiliation{%
  \institution{The University of Hong Kong}
  \city{Hong Kong}
  \country{China}
}
\email{zhiyi@cs.hku.hk}

\author{Zhihao Gavin Tang}
\affiliation{%
  \institution{Shanghai University of Finance and Economics}
  \city{Shanghai}
  \country{China}
}
\email{tang.zhihao@mail.shufe.edu.cn}

\author{Xiaowei Wu}
\affiliation{%
  \institution{University of Macau}
  \city{Macau}
  \country{China}
}
\email{xiaoweiwu@um.edu.mo}

\author{Yuhao Zhang}
\affiliation{%
 \institution{Shanghai Jiao Tong University}
 \city{Shanghai}
 \country{China}
}
\email{zhang_yuhao@sjtu.edu.cn}

%%
%% By default, the full list of authors will be used in the page
%% headers. Often, this list is too long, and will overlap
%% other information printed in the page headers. This command allows
%% the author to define a more concise list
%% of authors' names for this purpose.
% \renewcommand{\shortauthors}{Trovato et al.}

%%
%% The abstract is a short summary of the work to be presented in the
%% article.
\begin{abstract}

    This paper studies fractional matching on general graphs in the fully online model of Huang et al. (JACM 2020), in which all vertices arrive online and remain available for only a limited time. The algorithm must make irrevocable fractional matching decisions while the relevant vertices are simultaneously available.
    We extend the classic Water-Filling algorithm, also known as Balance and originally introduced by Kalyanasundaram and Pruhs (TCS 2000), to the fully online setting. Using an online primal-dual framework, we prove that the generalized Water-Filling algorithm achieves a competitive ratio of $2-\sqrt{2}\approx 0.586$ in the fully online model, and that this analysis is tight.
    To surpass the $2-\sqrt{2}$ barrier, we incorporate the ideas of eager matching and history-based pricing into Water-Filling. We show that the resulting algorithm achieves an improved competitive ratio of $0.599$, thereby establishing that Water-Filling is not optimal in the fully online setting.
    On the hardness side, we further improve the known upper bound for fractional fully online matching, reducing the previous best bound of $0.6297$ due to Eckl et al. (ORL 2021) to $0.6132$.
\end{abstract}

%%
%% The code below is generated by the tool at http://dl.acm.org/ccs.cfm.
%% Please copy and paste the code instead of the example below.
%%


%%
%% Keywords. The author(s) should pick words that accurately describe
%% the work being presented. Separate the keywords with commas.

% \received{20 February 2007}
% \received[revised]{12 March 2009}
% \received[accepted]{5 June 2009}

%%
%% This command processes the author and affiliation and title
%% information and builds the first part of the formatted document.
\maketitle

\renewcommand{\theHlemma}{\thesection.\arabic{lemma}}

% \tableofcontents

\newcommand{\pw}{p}
\newcommand{\eqdef}{\stackrel{\textrm{def}}{=}}
\newcommand{\diff}{\mathop{}\!\mathrm{d}}
\newcommand{\expect}[2]{\operatorname{\mathbf E}_{#1}\left[#2\right]}
\newcommand{\vect}[1]{\ensuremath{\vec{#1}}}
\newcommand{\vecy}{\vect{y}}
\newcommand{\vecmv}[1][v]{\vect{y}_{\text{-}#1}}
\newcommand{\onev}{\ensuremath{\mathsf{1}}}
\newcommand{\idr}[1]{\onev\left(#1\right)}
\newcommand{\var}{\mathsf{var}}
\newcommand{\opt}{\textsc{Opt}\xspace}
\newcommand{\alg}{\textsc{Alg}\xspace}
\newcommand{\dd}{\mathop{}\!\mathrm{d}}

\newcommand{\xiaowei}[1]{\textcolor{red}{(#1)}}
\newcommand{\zhihao}[1]{\textcolor{blue}{(#1)}}
\newcommand{\yuhao}[1]{\textcolor{yellow}{(#1)}}
\newcommand{\zhiyi}[1]{\textcolor{green}{(#1)}}
\newenvironment{suggestedrevision}[1][Suggested revision]
{%
    \par\smallskip
    \begingroup
    \color{blue}
    \noindent\textbf{#1.}\ \ignorespaces
}
{%
    \par
    \endgroup
    \smallskip
}
\newcommand{\ewf}{\textsc{Eager} \textsc{Water-Filling}\xspace}
\newcommand{\wtf}{\textsc{Water-Filling}\xspace}
\newcommand{\WTF}{\textsc{Water-Filling}\xspace}
\newcommand{\balance}{\textsc{Balance}\xspace}
\newcommand{\greedy}{\textsc{Greedy}\xspace}
\newcommand{\ranking}{\textsc{Ranking}\xspace}

\section{Introduction}

Matching theory is a cornerstone of combinatorial optimization, renowned for its elegant mathematical structure and foundational techniques.
In this field, the online matching problem has been a central theme in the online algorithms literature since the seminal work of \citet{stoc/KarpVV90}.
Their model considers the classic one-sided online matching problem, which involves a known set of offline vertices and a sequence of online vertices that arrive sequentially.
The algorithm must make an immediate and irrevocable matching decision upon each vertex's arrival.

A significant branch of this research focuses on the fractional online matching problem.
In this setting, the algorithm is permitted to make fractional assignments: for example, when an online vertex arrives and is adjacent to three available offline neighbors, the algorithm may assign a $1/3$ matched portion to each of the three edges.
This fractional variant is practically motivated by the large-volume perspective in real-world online matching applications, where each vertex may represent a type of resource or request with a large volume, and each atomic matching decision only affects a small fraction.
From a theoretical perspective, compared with the original integral version of the problem, the online algorithm can avoid the imbalances in matching decisions caused by integral constraints.
This allows us to clearly capture a core algorithmic question in the fractional version: how to make the most balanced matching decisions to optimize robustness against future uncertainty?

In the classic one-sided arrival setting, this question has been well studied.
\citet{tcs/KalyanasundaramP00} proposed the \wtf algorithm, also known as \balance, which achieves the optimal competitive ratio of $1-1/e$.
The balancing rule in this algorithm always chooses an available neighbor with the lowest water level, i.e., the smallest already matched fraction, thereby mimicking the physical water-filling process.
The same water-filling principle has been successfully extended to several weighted settings, including the vertex-weighted setting \cite{soda/AggarwalGKM11}, the edge-weighted setting \cite{wine/FeldmanKMMP09}, and the \textit{AdWords} setting \cite{jacm/MehtaSVV07}.
Within a broad class of online divisible-resource allocation problems, Water-Filling was recently shown to be universally minimax optimal for a large family of balancing objectives \cite{DBLP:journals/corr/abs-2603-26893}.
All these weighted extensions define a proper version of the water level and, by following the water-filling principle, successfully achieve the optimal competitive ratio of $1-1/e$.

The online bipartite matching was generalized to a fully online model by \citet{jacm/HuangKTWZZ20}, in which all vertices of the graph arrive online, and each vertex is available only during a specified time window from its arrival time to its departure deadline.
Upon the arrival of an online vertex, its incident edges with the existing vertices are revealed.
The model is motivated by the fact that the classic one-sided arrival model does not fully capture the dynamic nature of many real-world e-commerce applications. A typical example is online ride-hailing platforms, such as Uber, DiDi, and Lyft, where both drivers and passengers arrive online and remain in the system for a period of time while waiting to be matched. This is naturally captured by the active-window structure in the fully online model. Furthermore, in applications where matching decisions are made among users of the same type, such as passenger-to-passenger ride-sharing, the underlying graph becomes non-bipartite. This naturally raises the following central open question:

\begin{center}
\textsf{In the fractional fully online problem, is the water-filling idea still the optimal way to balance different choices \\
against the uncertain future? }
\end{center}

To address this, we first clarify a proper way to utilize the water-filling idea within the fully online matching problem.
It is straightforward to observe that a \emph{lazy} algorithm, which only makes matching decisions at a vertex's deadline (the final moment of its active time window), dominates other approaches.
Therefore, we naturally exploit the water-filling idea at each vertex's deadline and formally denote this approach as the \wtf algorithm for the fully online model. Our first result is to establish the tight competitive ratio of the \wtf algorithm.

\begin{theorem}
    \label{thm:wtf}
    The \wtf algorithm is $(2-\sqrt{2})$-competitive for fractional fully online matching on general graphs, and the ratio is tight for \wtf.
\end{theorem}

The next question is whether the competitive ratio $2-\sqrt{2}$ is the best possible for fractional algorithms in general. The hard instance for \wtf suggests that an algorithm may be able to perform better if it uses the \emph{inter-available} structure when balancing its matching decisions. Here, the inter-available structure refers to the subgraph induced by vertices that are still available, i.e., vertices whose active time windows have not expired. Edges in this subgraph are not currently expiring: no endpoint has reached its deadline, and hence the algorithm is not forced to decide their matching values immediately.

This structure is special to the fully online model. It stands in contrast to \emph{expiring} structures, which consist of edges for which at least one endpoint has reached its deadline. In the classic one-sided arrival model, matching decisions are forced upon the arrival of an online vertex. Therefore, every observed decision is effectively an expiring decision. It can be shown that the water level is a simple yet optimal metric for capturing such expiring structures, which leads to the optimal $1-1/e$ competitive ratio in the one-sided model. In fact, even in the fully online model, the competitive ratio $2-\sqrt{2}$ remains an upper bound for all \emph{inter-available-independent} algorithms, namely algorithms whose matching decisions do not depend on the inter-available graph structure. Thus, the central theoretical question is:
\begin{center}
    \textsf{Can a new algorithmic approach exploit the inter-available structure to achieve a better competitive ratio \\
    in the fully online model?}
\end{center}

In our next contribution, we explore the concept of \emph{eager matching} (matching upon arrival) to capture the inter-available structure, and we link this approach to competitive analysis via an economic interpretation of matching algorithm design, grounded in the online primal-dual framework.
Specifically, we first interpret the \wtf algorithm from an economic perspective by defining a price function for each vertex relative to its water level.
Note that this price corresponds to the dual assignment for the vertex, thereby bridging the algorithm's logic to the primal-dual analysis.

Under this framework, the remaining task is to design the timing of matching decisions, where each vertex makes its matching decision by selecting the neighbor with the lowest current price.
For instance, the \wtf algorithm utilizes a monotone, water-level-based price function and requires all vertices to make continuous matching decisions only at their respective deadlines.

Consequently, we observe that it becomes advantageous to perform \emph{eager} matching at the time of a vertex's arrival under certain conditions.
We remark that this type of matching occurs precisely within the inter-available structure, as both the vertex and its neighbor remain available after the matching occurs. This contrasts with the \emph{lazy} matching at a deadline, where at least one vertex departs the system immediately following the decision.
This leads to the \ewf algorithm, which introduces an additional matching opportunity for each vertex, and to the following theorem.

\begin{restatable}{theorem}{thmEagerWf}
    \label{thm:eager-wf}
    The \ewf algorithm achieves a competitive ratio of $0.592$ for fractional fully online matching on general graphs.
\end{restatable}

We further show that the algorithm and competitive ratio can be improved via a more fine-grained design of the price function from an economic viewpoint.
While both the \wtf and \ewf algorithms utilize purely water-level-based pricing, this approach is no longer optimal once eager matching is introduced to exploit the inter-available structure.
More fine-grained historical information becomes available to identify neighbors' states at the moment of a vertex's matching decision.
For instance, two neighbors with identical water levels may still differ as follows: one may have acquired its water level actively upon arrival through eager matching, while the other accumulated its water level passively at the deadlines of other vertices. In contrast to the \wtf algorithm, where all water levels are accumulated passively, this distinction becomes crucial.
To address this, we design a history-based pricing scheme to capture the different channels of accumulated water levels, thereby establishing an enhanced version of the \ewf algorithm.
As a result, we further improve the competitive ratio to $0.599$.

\begin{restatable}{theorem}{thmEagerWfHistory}
    \label{thm:eager-history}
	The \ewf algorithm with history-based pricing achieves a competitive ratio of $0.599$ for fractional fully online matching on general graphs.
\end{restatable}

\paragraph{Hardness Results}
The final contribution of this paper is an improved hardness result for fractional fully online matching. Recall that any randomized integral algorithm can be simulated by a deterministic fractional algorithm whose fractional matching value equals the expected matching value of the randomized algorithm. Therefore, any impossibility result for the fractional setting immediately implies the same impossibility result for randomized integral algorithms. Almost all previous upper bounds in the online matching literature were established in the fractional model. For example, the classic $1-1/e$ upper bound in the one-sided arrival model was proved for fractional algorithms. The same is true for the $0.6317$ upper bound for the fully online model, first shown by \citet{jacm/HuangKTWZZ20} and later improved to $0.6297$ by \citet{orl/EcklKLS21}. In particular, $0.6297$ was the best previously known upper bound for fully online matching, both for fractional algorithms and for randomized integral algorithms.

We improve this upper bound to $0.6132$. Together with our positive results, this narrows the optimal competitive ratio for fractional fully online matching to the interval $[0.599,0.6132]$.

\begin{restatable}{theorem}{fullyUpper}
\label{thm:fullyupper}
No algorithm can achieve a competitive ratio strictly better than $0.6132$ for the fractional fully online matching problem, even on bipartite graphs.
\end{restatable}

Our main technical contribution is a new hard-instance construction that is qualitatively different from previous ones. Earlier constructions embed a growing tree into the standard upper-triangle instance, which is the hard instance used in the $1-1/e$ upper bound in the one-sided model. The role of the growing tree is to create a worse initial state for the upper-triangle instance, thereby pushing the hardness below $1-1/e$.

In contrast, our construction is based on a repeated pattern, reminiscent of the tight hard instances for \textsc{Ranking}~\cite{soda/HuangPTTWZ19} and \wtf\ (\Cref{sec:wtf}). Each repetition consists of two stages: an upper-triangle stage followed by a complete bipartite stage. The second stage amplifies the mistakes made by the algorithm in the first stage. Moreover, by repeating this pattern many times, we can drive the upper bound down repeatedly, rather than obtaining only a one-shot loss as in previous constructions.

\subsection{Further Related Work}

Online bipartite matching is a central problem in the design and analysis of online algorithms, initiated by the seminal work of \citet{stoc/KarpVV90}. They showed that the natural \greedy algorithm achieves a competitive ratio of $0.5$, since it always produces a maximal matching, and that this ratio is optimal among deterministic integral algorithms. In contrast, by relaxing the integrality constraint and allowing randomization, they introduced the randomized algorithm called \ranking, which achieves the optimal competitive ratio of $1-1/e$. The analysis of \ranking was later simplified in several follow-up works~\cite{soda/GoelM08, sigact/BenjaminM08, soda/DevanurJK13}.

Since then, many variants of online matching have been studied from two main perspectives: generalized arrival models and weighted objectives. On the arrival side, \citet{icalp/WangW15}, \citet{focs/GamlathKMSW19}, and \citet{DBLP:conf/sigecom/Tang024} considered the general vertex-arrival model, in which all vertices arrive online and irrevocable matching decisions must be made upon arrival. The best competitive ratio currently known for fractional algorithms in this model is $0.526$~\cite{icalp/WangW15,DBLP:conf/sigecom/Tang024}, and this ratio was shown to be tight by \citet{DBLP:journals/corr/abs-2602-18049}. For randomized integral algorithms, the current best competitive ratio is $0.5+\varepsilon$, for some small constant $\varepsilon$, achieved by \citet{focs/GamlathKMSW19}, thereby breaking the $0.5$ barrier attained by \greedy.

The fully online model studied in this paper is another natural way to capture the setting in which all vertices arrive online. In this model, \citet{jacm/HuangKTWZZ20} showed that \ranking achieves a competitive ratio of $0.567$ on bipartite graphs, which is tight for \ranking, and $0.5211$ on general graphs. Subsequently, \citet{focs/HuangTWZ20} gave an improved algorithm with a competitive ratio of $0.569$ on bipartite graphs, while \citet{derakhshan2026unified} obtained an improved ratio of $0.539$ on general graphs. At the other extreme, if one considers the more general edge-arrival model, then \citet{focs/GamlathKMSW19} proved that no fractional algorithm can achieve a competitive ratio strictly better than $0.5$. Furthermore, \citet{ec/AshlagiBDJSS19} considered the edge-weighted version of the fully online setting, under the fixed-length time window assumption.

For the online bipartite matching with weighted objective, \citet{soda/AggarwalGKM11} studied the vertex-weighted variant, in which offline vertices carry weights, and obtained the same optimal $1-1/e$ competitive ratio as in the unweighted setting. For edge-weighted online matching, the most direct formulation does not admit any bounded competitive ratio, which has led to two alternative models: \emph{AdWords} (edge-weighted matching with budgets) and \emph{Display Ads} (edge-weighted matching with free disposal). For AdWords, \citet{jacm/MehtaSVV07} achieved a competitive ratio of $1-1/e$ using a fractional algorithm, or under the assumption that all bids are small. Without the small-bids assumption, the first result to beat the $0.5$ barrier was obtained by \citet{focs/HuangZZ20}, and the current best ratio is $0.504$ by \citet{DBLP:conf/focs/Chen0S24}. The Display Ads problem was first studied by \citet{wine/FeldmanKMMP09}, who proposed a fractional algorithm achieving a competitive ratio of $1-1/e$, again under the small-bids assumption. The first randomized integral algorithm to beat the $0.5$ barrier was given by \citet{DBLP:journals/jacm/FahrbachHTZ22}, based on a novel online correlated selection technique. Several subsequent works improved this line of research \cite{DBLP:conf/isaac/ShinA21, DBLP:conf/focs/Chen0S24, DBLP:conf/focs/BlancC21}, and the current best ratio is $0.5368$, achieved by \citet{DBLP:conf/focs/BlancC21}.

\section{Preliminaries}
\label{sec:pre}
We first introduce the notation and formal definitions used throughout the paper.

In the \textbf{fully online matching problem}, there is an underlying graph $G=(V,E)$ which is not necessarily bipartite. Vertices arrive and depart over time. For each vertex $v \in V$, let $r_v$ and $d_v$ denote its arrival time and deadline. We call $[r_v,d_v]$ the \emph{available} time window of $v$, and say that $v$ is available at time $t$ if $t \in [r_v,d_v]$. We assume that two vertices can be adjacent only if their available windows overlap, i.e., $(u,v)\in E$ only if $[r_u,d_u]$ and $[r_v,d_v]$ intersect.

The goal is to design an online algorithm that incrementally constructs a matching in $G$, with all decisions being irrevocable. The graph $G$ is not revealed in advance. When a vertex $v$ arrives at time $r_v$, the algorithm learns of $v$ together with all edges connecting $v$ to previously arrived vertices. At any time, if two active vertices $u$ and $v$ are both unmatched, the algorithm may irrevocably match them along the edge $(u,v)$, provided that this edge exists. Moreover, when a vertex $v$ departs at its deadline $d_v$, the algorithm is notified that this is the last opportunity to match $v$.

Then, we present the standard LP relaxation for matching (both primal and dual), together with the fractional version of the fully online matching problem.
\begin{align*}
\textsf{(Primal) } \max: \quad & \textstyle \sum_{(u,v)\in E} x_{uv} && \qquad\qquad & \textsf{(Dual) }\min: \quad & \textstyle\sum_{u \in V} \alpha_u\\
\text{s.t.} \quad & \textstyle \sum_{v:(u,v)\in E} x_{uv} \leq 1 && \forall u\in V & \text{s.t.} \quad & \alpha_u + \alpha_v \geq 1 && \forall (u,v)\in E \\
& x_{uv} \geq 0 && \forall (u,v)\in E & & \alpha_u \geq 0 && \forall u \in V
\end{align*}

\textbf{The fractional fully online matching problem} reduces the task of producing an integral matching of the graph to maintaining a feasible fractional solution of the (Primal) LP. Equivalently, we may fractionally match an edge $(u,v)$ with some $x_{uv} \leq 1$, while ensuring that for each vertex $u$, the total fractional matched portion 
$
\sum_{v:(u,v)\in E} x_{uv} \;\leq\; 1.
$
During the process, the fractional algorithm may only \emph{increase} $x_{uv}$ for an existing edge $(u,v)$, and only when both endpoints $u$ and $v$ are active.

All the algorithms discussed in this paper are analyzed under the \textbf{primal-dual framework}. In particular, we maintain a (possibly infeasible) dual solution $\alpha$ together with the fractional matching (a feasible primal solution $x$), and prove the following properties:
\begin{itemize}
    \item (Equal Update) $\sum_{(u,v)\in E} x_{uv} \;=\; \sum_{u \in V} \alpha_u$. 
    \item (Approximate Dual Feasibility) $\alpha_u \geq 0$ for all $u \in V$, and $\alpha_u + \alpha_v \geq \Gamma$ for all $(u,v) \in E$.
\end{itemize}

From these properties, weak duality directly implies that the feasible primal solution produced is at least a $\Gamma$-fraction of the optimal fractional solution, and hence the fractional algorithm is $\Gamma$-competitive. Specifically, since $\alpha$ is $\Gamma$-approximately feasible, the scaled vector $\alpha/\Gamma$ is a feasible dual solution. By weak duality, letting $x^*$ denote the optimal fractional solution, we have
\[
\sum_{(u,v)\in E} x^*_{uv} \;\leq\; \sum_{u \in V} \frac{\alpha_u}{\Gamma} 
\;=\; \frac{1}{\Gamma} \sum_{u \in V} \alpha_u 
\;=\; \frac{1}{\Gamma} \sum_{(u,v)\in E} x_{uv}.
\]
Formally, we have the following lemma.
\begin{lemma}
    \label{lem:primal-dual}
    Consider a fractional algorithm that maintains a feasible fractional solution $x$. 
    If there exists a nonnegative, possibly infeasible dual vector $\alpha$ that can be maintained 
    simultaneously and satisfies the \textnormal{(Equal Update)} and 
    \textnormal{(Approximate Dual Feasibility with factor $\Gamma>0$)} properties, then the fractional algorithm is $\Gamma$-competitive.
\end{lemma}

\section{Tight Analysis for the Water-Filling Algorithm}
\label{sec:wtf}
In this section, we examine the standard \wtf algorithm (also known as the \balance algorithm) for the fractional fully online matching problem. This algorithm was originally proposed in the one-sided online bipartite matching model, and we adapt it to the fully online setting. The high-level idea is to continuously allocate a vertex to its least-matched neighbors, with the process initiated at the vertex's deadline. A formal description is provided in \Cref{alg:wtf}, along with the maintenance of the dual variables $\alpha$. Note that in the algorithm, we use $x_u \eqdef \sum_{v:(u,v)\in E}x_{uv}$ to keep track of the water level (i.e., total fractional mass) of $u$ at all times, and $g:[0,1]\to[0,1]$ is a fixed increasing function for dual maintenance, which will be specified later.

\begin{algorithm}[ht]
\caption{The \wtf Algorithm}
\label{alg:wtf}
% \DontPrintSemicolon
\SetAlgoLined
\KwData{Graph $G=(V,E)$ with arrivals and deadlines on vertices;}
\KwResult{Fractional matching $\{x_{uv}\}$ and dual variables $\{\alpha_u\}$;}
\SetKwIF{When}{ElseWhen}{Else}{Trigger: when}{do}{Else when}{else}{end}
\textbf{Initialization:} Initialize each variable only once: when a vertex $u$ arrives, set $x_u\gets 0$ and $\alpha_u\gets 0$; when an edge $(u,v)$ is revealed for the first time, set $x_{uv}\gets 0$.\;

\When{a vertex $u$ departs}{
  $\pw_u \gets x_u$ \tcp*{passive water level collected before $u$'s deadline}
  $N(u)\gets \{\,v:\ v\text{ is a neighbor of }u\text{ and currently available}\,\}$ \tcp*{available neighbor set}

  \While{$x_u < 1$ and $\min_{w \in N(u)} x_w < 1$}{
    $v \gets \arg\min_{w \in N(u)} x_w$\;
    $x_{uv} \gets x_{uv} + \diff x$ \tcp*{match a small amount $\diff x$ on $(u,v)$}
    $x_u \gets x_u + \diff x, \quad x_v \gets x_v + \diff x$ \tcp*{water-level updates}
    $\diff \alpha_u \gets (1 - g(x_v))\,\diff x$, \quad $\diff \alpha_v \gets g(x_v)\,\diff x$ \tcp*{dual updates}
    \Indm
  }
}
\end{algorithm}

\subsection{Lower Bound on the Competitive Ratio}
\label{sec:wtf_lb}

We now turn to the analysis of the \wtf algorithm. We show, via the primal-dual framework, that its competitive ratio is at least $2-\sqrt{2}$.

\begin{lemma}[Positive Part of \Cref{thm:wtf}]
	\label{thm:wtf_lower}
	The \wtf algorithm is $(2-\sqrt{2})$-competitive.
\end{lemma}
\begin{proof}
We employ the primal-dual framework defined in \Cref{sec:pre}. The algorithm maintains the primal and dual variables such that the Equal Update property holds (i.e., the total objective values increase at the same rate). By Lemma~\ref{lem:primal-dual}, it suffices to prove that the dual variables are approximately feasible: for every pair of neighbors $u$ and $v$, $\alpha_u + \alpha_v \ge 2-\sqrt{2}$. Note that in our analysis, $g(x)$ is set to be $\tfrac{\sqrt{2}}{2}x + 1 - \tfrac{\sqrt{2}}{2}$.

For each vertex $u$, the dual variable $\alpha_u$ increases alongside $x_u$ in two situations:
\begin{itemize}
    \item \textbf{Active:} At $u$'s deadline, $x_u$ increases by $\diff x$, and $\alpha_u$ increases by $(1 - g(x_v))\,\diff x$, where $v$ is the neighbor matched by $u$ at this infinitesimal step.
    \item \textbf{Passive:} At the deadline of another vertex $v$, $x_u$ increases by $\diff x$, and $\alpha_u$ increases by $g(x_u)\,\diff x$.
\end{itemize}

Fix any pair of neighboring vertices $u, v$ where $u$ has an earlier deadline than $v$. Consider the moment immediately \textbf{after} $u$'s deadline (i.e., after the \wtf algorithm completes the matching of $u$). At this point, at least one of $x_u = 1$ and $x_v = 1$ holds; otherwise, $x_u$ would continue to increase. We analyze two cases separately:

\paragraph{Case 1: $x_v = 1$.}
    Since $v$ has not yet reached its deadline, its water level and dual value are derived entirely from passive increments. Therefore, following the dual update rule, we can lower-bound the sum $\alpha_u + \alpha_v$ by considering $\alpha_v$ alone:
    \[
        \alpha_u + \alpha_v \ge \alpha_v \ge \int_0^1 g(x)\,\diff x.
    \]
    Using the function $g(x) = \frac{\sqrt{2}}{2}x + (1 - \frac{\sqrt{2}}{2})$, we have $\int_0^1 g(x)\,\diff x = 1 - \frac{\sqrt{2}}{4} \ge 2-\sqrt{2}$. Thus, the bound holds.

    \paragraph{Case 2: $x_v < 1$ (which implies $x_u = 1$).}
    Again, since $v$ has not reached its deadline, its gain is purely passive:
    \[
        \alpha_v = \int_0^{x_v} g(x)\,\diff x.
    \]
    For $\alpha_u$, the dual gain consists of two components. Prior to $u$'s deadline, $u$ accumulates value passively until its water level reaches $\pw_u$ (the passive water level of $u$, defined in \Cref{alg:wtf}):
    \[
        \text{Gain}_{\text{passive}} = \int_0^{\pw_u} g(x)\,\diff x.
    \]
   Subsequently, during the active matching procedure at $u$'s deadline, $x_u$ increases from $\pw_u$ to $1$. The rate of increase of $\alpha_u$ depends on the water level of the neighbor matched to $u$. We rely on a crucial property maintained throughout the process:
    \[
        \textbf{The water level of any neighbor matched to $u$ is at most $x_v$}.
    \]
    This holds because $v$ is a valid available neighbor with water level at most $x_v$, and the \wtf algorithm greedily selects the neighbor with the lowest water level. Consequently, we can lower-bound the active gain as follows:
    \[
        \text{Gain}_{\text{active}} \ge \int_{\pw_u}^1 \bigl(1-g(x_v)\bigr)\,\diff x = (1-\pw_u)\bigl(1-g(x_v)\bigr).
    \]
    Combining the bounds for $\alpha_u$ and $\alpha_v$:
    \begin{equation}
        \label{eqn:wtf_gain}
        \alpha_u + \alpha_v \ge \underbrace{\int_0^{\pw_u} g(x)\,\diff x}_{\text{Gain}_{\text{passive}}} + \underbrace{(1-\pw_u)(1-g(x_v))}_{\text{Gain}_{\text{active}}} + \underbrace{\int_0^{x_v} g(x)\,\diff x}_{\alpha_v}.
    \end{equation}
    Substituting $g(x) = \frac{\sqrt{2}}{2}x + (1 - \frac{\sqrt{2}}{2})$:
    \begin{align*}
        \alpha_u + \alpha_v &\ge \int_0^{\pw_u} g(x)\,\diff x + (1-\pw_u)(1-g(x_v)) + \int_0^{x_v} g(x)\,\diff x \\
        &= \left(\tfrac{\sqrt{2}}{4}\pw_u^2 + (1-\tfrac{\sqrt{2}}{2})\pw_u\right) + \left((1-\pw_u)\tfrac{\sqrt{2}}{2}(1-x_v)\right) + \left(\tfrac{\sqrt{2}}{4}x_v^2 + (1-\tfrac{\sqrt{2}}{2})x_v\right) \\
        &= \tfrac{\sqrt{2}}{4}\left((\pw_u + x_v) - (2-\sqrt{2})\right)^2 + 2-\sqrt{2} \\
        &\ge 2-\sqrt{2}.
    \end{align*}
    By Lemma~\ref{lem:primal-dual}, this implies a competitive ratio of at least $2-\sqrt{2}$.
\end{proof}

\subsection{Hardness Results on the Competitive Ratio}
\label{sec:wtf_ub}

We then provide a matching upper bound showing that the competitive ratio of $2-\sqrt{2}$ is tight for \wtf. In fact, although we only analyze the performance of \wtf on the hard instance, the construction applies to a broader family of algorithms called inter-available-independent algorithms. An algorithm is called inter-available-independent if it ignores edges between two available vertices that have not yet reached their deadlines. Equivalently, we may assume that such edges are revealed to the algorithm only when one of their endpoints reaches its deadline. Clearly, \wtf is inter-available-independent, for two reasons: (1) it is a lazy algorithm that makes matching decisions only at deadlines; and (2) its decisions depend only on the water levels of vertices, which are determined solely by expired edges, that is, edges for which at least one endpoint has already passed its deadline.

\paragraph{Hard Instance.}
There are $2km$ vertices partitioned into $m$ groups, each of size $2k$.
For each $t\in[m]$, the $t$-th group is $U_t\cup V_t$, where
$U_t=\{u_{t,1},\ldots,u_{t,k}\}$ and $V_t=\{v_{t,1},\ldots,v_{t,k}\}$.
Let $h:[0,1]\to[0,1]$ be a decreasing function (to be specified later) with $h(0)=1$ and $h(1)=0$.\footnote{When $h(x)\equiv 1$, the instance reduces to the $\frac{1}{1+\ln 2}\approx 0.5906$ hard instance of~\cite{stacs/EpsteinLSW13} for the edge-arrival model.}
The edge set has two types (see Figure~\ref{figure:subgraph_induced}):
\begin{enumerate}
  \item \emph{Upper-triangular edges} between $U_t$ and $V_t$: for all $t\in[m]$, $i\in[k]$, and $j\ge i$, include $(u_{t,i},v_{t,j})\in E$;
  \item \emph{$h$-induced edges} between $U_t$ and $U_{t+1}$: for all $t\in[m-1]$ and $i\in[k]$, for every $j\le \lfloor k\,h(\frac{i-1}{k})\rfloor$, include $(u_{t,i},u_{t+1,j})\in E$.
\end{enumerate}

\begin{figure}[H]
	\centering
	\resizebox*{65mm}{!}
	{
	\begin{tikzpicture}
	[font=\bf,line width=0.3mm,scale=.8,auto=left,every node/.style={circle,draw}]
	\def \n {4}
	\foreach \i in {1,...,\n}
	{
		\node at (-4,1.4*\n-1.4*\i) (w\i) {$u_{2,\i}$};
	}
	\foreach \i in {1,...,\n}
	{
		\node at (1,1.4*\n-1.4*\i) (u\i) {$u_{1,\i}$};
	}
	\foreach \i in {1,...,\n}
	{
		\node at (6,1.4*\n-1.4*\i) (v\i) {$v_{1,\i}$};
		\foreach \j in {1,...,\i}
		{
			\draw (u\j) to (v\i);
		}
	}
	\draw (u1) to (w1);
	\draw (u1) to (w2);
	\draw (u1) to (w3);
	\draw (u1) to (w4);
	\draw (u2) to (w1);
	\draw (u2) to (w2);
	\draw (u2) to (w3);
	\draw (u3) to (w1);
	\draw (u3) to (w2);
	\draw (u3) to (w3);
	\draw (u4) to (w1);
	\end{tikzpicture}
	}
	\caption{Subgraph induced by $U_t\cup V_t \cup U_{t+1}$: illustrating example with $t=1$ and $k=4$.}
    \Description{It is an example with $t=1$ and $k=4$.}
	\label{figure:subgraph_induced}
\end{figure}
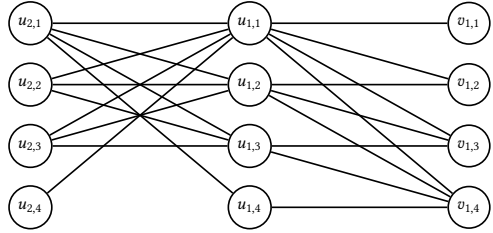

Finally, the deadlines of the $u$-vertices are reached first, following the lexicographic order on $(t,i)$.
After all $u$-vertices have reached their deadlines (i.e., after $u_{m,k}$), the deadlines of the $v$-vertices occur.
The relative order among the $v$-vertices does not matter, as long as each $v_{t,i}$'s deadline follows that of $u_{t,i}$. For concreteness, let all vertices arrive, in an arbitrary fixed order, before the first $u$-vertex reaches its deadline.  Thus every edge is revealed before any matching decision of \wtf, all adjacent time windows overlap, and the deadlines are exactly those specified above. Inter-available-independent algorithms ignore the additionally revealed inter-available structure by definition.

It is easy to verify that this hard instance is bipartite, with
$(U_1\cup U_3\cup \dots)\cup (V_2\cup V_4\cup \dots)$ and
$(U_2\cup U_4\cup \dots)\cup (V_1\cup V_3\cup \dots)$
forming the two vertex partitions.
This graph admits a perfect matching in which $u_{t,i}$ is matched to $v_{t,i}$ for all $t\in[m]$ and $i\in[k]$, and hence $\opt = km$.

We next construct the function $h$ and establish the following technical lemma.
Let $c = 2-\sqrt{2}$, and define a function $f:[0,c]\to[0,1]$ by
\begin{align*}
f(x) \;\eqdef\; \tfrac{1}{2}\bigl(\ln(1-x)+\ln(1-c+x)\bigr)
+ \tfrac{1}{\sqrt{2}(x-1)}
+ \tfrac{2+\sqrt{2}-\ln(1-c)}{2}.
\end{align*}
It is straightforward to verify that $f$ is strictly decreasing.
Let $\tau(x) \eqdef f^{-1}(x)$ and define $h(x) \eqdef f(c - f^{-1}(x))$.
Since $f(0)=1$ and $f(c)=0$, the functions $h:[0,1]\to[0,1]$ and $\tau:[0,1]\to[0,c]$ are both well defined.
Moreover, $h$ is decreasing with $h(0)=1$ and $h(1)=0$, as required in the construction of the hard instance.
Although these functions may appear ad hoc at first glance, their design is not arbitrary:
in \Cref{sec:wtf_gh_dual}, we reveal a duality connection between $h$ and the gain-sharing function $g$ used to define the dual variables in \wtf.

%We first establishes three equations satisfied by these functions.

\begin{lemma}\label{lemma:property_of_h}
	For all $x\in [0,1]$ we have
	\begin{equation*}
	\int_0^x \frac{1-\tau(y)}{1-y+h(y)} dy = c-\tau(x),
	\quad \int_0^1 \tau(y)dy = 1-c, \quad \text{and} \quad \int_0^1 \frac{1}{1-y+h(y)} dy < 1.
	\end{equation*}
\end{lemma}
\begin{proof}
	First, we show the first equation, i.e., for all $x\in[0,1]$ we have $\int_0^x \frac{1-\tau(y)}{1-y+h(y)} dy = c-\tau(x)$.
	Note that $\tau(0) = c$ and, thus, both sides equal $0$ when $x = 0$.
	It suffices to check that for all $x\in[0,1]$, $\frac{1-\tau(x)}{1-x+h(x)} = -\tau'(x)$.
	Let $\phi = \tau(x)\in[0,c]$. Then $f(\phi) = x$ and $h(x) = f(c-f^{-1}(x)) = f(c-\phi)$.
	Then, we only need to check that
	\begin{equation*}
	\frac{1-\phi}{1-f(\phi)+f(c-\phi)} = -\tau'(x) = -\frac{1}{f'(\phi)},
	\end{equation*}
	which is true as $f$ is defined such that for all $\phi \in[0,c]$,
	\begin{equation*}
	1-f(\phi)+f(c-\phi)+(1-\phi)f'(\phi) = 0.
	\end{equation*}

	Integrating from $0$ to $c$, the contributions of the second and third terms cancel.
	We have
	\begin{equation*}
		0 = c+\int_0^c (1-x)f'(x) dx = c - 1 + \int_0^c f(x) dx,
	\end{equation*}
	which implies the second equation because $\int_0^1 \tau(y)dy = \int_0^c f(x) dx = 1-c$, where the first equality follows because $\tau = f^{-1}$, $f$ is strictly decreasing, $f(0) = 1$, and $f(c) = 0$.

	Now we prove the last equation, i.e., $\int_0^1 \frac{1}{1-y+h(y)} dy < 1$.

	Observe that both $1-\tau(y)$ and $\frac{1}{1-y+h(y)}$ are increasing in terms of $y$.
	Hence we have
	\begin{align*}
	c = c-\tau(1)
    &=  \int_0^1 \frac{1-\tau(y)}{1-y+h(y)} dy \\
	&>  \int_0^1 (1-\tau(y)) dy \cdot \int_0^1 \frac{1}{1-y+h(y)} dy \\
	&= c\cdot \int_0^1 \frac{1}{1-y+h(y)} dy.
	\end{align*}

	Dividing both sides by $c$ proves the last equation.
\end{proof}

Now we analyze the performance of \wtf on this instance.
We first prove that by running \wtf on the hard instance, the passive water levels of almost all vertices are strictly smaller than $1$.

\begin{lemma} \label{lemma:passive_<_1}
	For large enough $k$, \wtf produces a fractional matching with $p_{u_{t,i}} < 1$ for all $t \in [m],i\in[k]$ and $p_{v_{t,i}} < 1$ for all $t \in [m-1],i\in[k]$.
\end{lemma}
\begin{proof}
Observe that at the deadline of each $u_{t,i}$, where $t\in[m-1]$, it has $|N(u_{t,i}) \cap V_t| + |N(u_{t,i}) \cap U_{t+1}| = k-i+1+\lfloor k\cdot h(\frac{i-1}{k}) \rfloor$ neighbors whose deadlines have not yet been reached. Moreover, as $h$ is decreasing, it is easy to see (by induction) that at the deadline of $u_{t,i}$, all available neighbors of $u_{t,i}$ have the same water level.
Hence, \wtf increases the water level of the available neighbors of $u_{t,i}$ at the same rate until $\min_{v\in N(u_{t,i})} x_v =1$ or $x_{u_{t,i}}=1$.
Since $u_{t+1,1}$ is a neighbor of every vertex in $U_t$, we have $p_{u_{t+1,1}} = \max_{j\in[k]} \{ p_{u_{t+1, j}},p_{v_{t, j}} \}$.
Therefore, it suffices to show that $p_{u_{t+1,1}}$ is smaller than $1$.
Note that each vertex $u_{t,i}$ has at most $1$ unit of unmatched portion that is distributed among $k-i+1+\lfloor k\cdot h(\frac{i-1}{k})\rfloor$ available neighbors and, thus, it increases the water level of $u_{t+1,1}$ by at most $\frac{1}{k-i+1+\lfloor k\cdot h(\frac{i-1}{k})\rfloor}$.
Hence, when $k \to \infty$, we have
\begin{align*}
p_{u_{t+1,1}} \le \sum_{i=1}^{k} \frac{1}{k-i+1+\lfloor k\cdot h(\frac{i-1}{k}) \rfloor} \approx \int_0^1 \frac{1}{1-y+h(y)} dy < 1,
\end{align*}
where the last inequality follows from Lemma~\ref{lemma:property_of_h}. This finishes the proof.
\end{proof}

Lemma~\ref{lemma:passive_<_1} implies that, for large enough $k$, we can guarantee that when running \wtf on the hard instance, after the deadline of every $u_{t,i}$, where $t\in[m-1]$, we must have $x_{u_{t,i}} = 1$, as none of its neighbors with a later deadline has water level $1$.

\begin{corollary}\label{corollary:water-level-reaching-1}
	For all $t\in[m-1]$, we have $x_{u_{t,i}}=1$ after $u_{t,i}$'s deadline.
\end{corollary}
%\zhiyi{Phrase as a corollary for reference in the next proof?}

Now we are ready to prove the main theorem of this section.

\begin{lemma}[Hardness Part of \Cref{thm:wtf}] \label{th:wthard}
    The competitive ratio of \wtf is at most $2-\sqrt{2}$.
\end{lemma}
\begin{proof}
	Let $\mathbf{\pw}_t = (\pw_{u_{t,1}},\pw_{u_{t,2}},\ldots,\pw_{u_{t,k}})^\text{T}$ denote the passive water-level vector of $U_t$.
	Since the increment of matching at $u_{t,i}$'s deadline is at most $1-p_{u_{t,i}}$, the solution given by \wtf is
	\begin{equation*}
	\sum_{(u,v)\in E}x_{uv} \leq \sum_{t,i}(1-\pw_{u_{t,i}})
	= \sum_{t} (k-\|\mathbf{p}_t\|_1).
	\end{equation*}

	Indeed, by Corollary~\ref{corollary:water-level-reaching-1}, for all $t\in[m-1]$, the increment of matching at $u_{t,i}$'s deadline is exactly $1-p_{u_{t,i}}$.
	Recall that in the hard instance, $u_{t+1,i}$ is a neighbor of $u_{t,j}$ if and only if $\frac{i}{k} \leq h(\frac{j-1}{k})$.
	Hence we have
	\begin{align*}
	\pw_{u_{t+1,i}} = \sum_{j=1}^{\lfloor k\cdot h^{-1}(\frac{i}{k})+1 \rfloor} \frac{1-p_{u_{t,j}}}{k-j+1+\lfloor k\cdot h(\frac{j-1}{k}) \rfloor} 
	= \sum_{j=1}^{\lfloor k\cdot h^{-1}(\frac{i}{k})+1 \rfloor}(1-p_{u_{t,j}})\cdot a_j,
	\end{align*}
	where $a_j = \frac{1}{k-j+1+\lfloor k\cdot h(\frac{j-1}{k}) \rfloor}$ is independent of $t$.
	In other words, there exists a $k\times k$ matrix $\mathsf{M}$ such that for all $t\in[m-1]$, $\mathbf{\pw}_{t+1} = \mathsf{M}(\mathbf{1}-\mathbf{\pw}_t)$.
	More precisely, we have $\mathsf{M}_{i,j} = a_j$ if $j\leq \lfloor k\cdot h^{-1}(\frac{i}{k})+1 \rfloor$, $\mathsf{M}_{i,j} = 0$ otherwise.
	Hence, for any $i\in[k]$, by Lemma~\ref{lemma:property_of_h}, we have
	\begin{equation*}
	\sum_{j\in [k]}\mathsf{M}_{i,j} \leq \sum_{j\in [k]}a_j < 1.
	\end{equation*}

	That is, $\mathsf{M}$ is a contraction matrix and the above mapping from $\mathbf{p}_{t}$ to $\mathbf{p}_{t+1}$ has a unique stationary vector $\mathbf{p^*}$, i.e. $\mathbf{p^*}= \mathsf{M}(\mathbf{1}-\mathbf{p^*})$. Moreover, $\lim_{t \to \infty}\mathbf{\pw}_t = \mathbf{p^*}$\footnote{Observe that $(\mathbf{\pw}_{t+1}-\mathbf{p^*})= \mathsf{M}(\mathbf{p^*}-\mathbf{\pw}_t)$ and $\mathsf{M}$ is a contraction matrix.}. Thus, for any fixed $k$, when $m\to \infty$, the ratio between the matching size of \wtf and the optimum is
	\[
	\lim_{m\to \infty} \frac{\sum_{t} (k-\|\mathbf{p}_t\|_1)}{mk} = 1-\frac{1}{k}\cdot\|\mathbf{p^*}\|_1.
	\]

	Finally, we consider when $k\to \infty$ and calculate the stationary vector.
	In this case, $\mathbf{p^*}$ becomes a function $p:[0,1]\to[0,1]$ and the linear equation $\mathbf{p^*}= \mathsf{M}(\mathbf{1}-\mathbf{p^*})$ becomes the following
	\[
	\int_0^{h^{-1}(x)} \frac{1-p(y)}{1-y+h(y)} dy = p(x), \quad \forall x \in [0,1].
	\]

	We verify that $p=\tau$ is a solution to this system of equations by Lemma~\ref{lemma:property_of_h}. For all $x$, we have
	\begin{align*}
		\int_0^{h^{-1}(x)} \frac{1-\tau(y)}{1-y+h(y)} dy &=  c-\tau(h^{-1}(x)) \\
		& = \tau\left( f\left(c-\tau(h^{-1}(x))\right) \right) \\
		& = \tau\left(h\left(h^{-1}(x)\right)\right) = \tau(x).
	\end{align*}

	Thus, the ratio between \wtf and \opt is $1-\int_0^1 \tau(y) dy = c = 2 - \sqrt{2}$.
	%,
	%which proves the claim.
\end{proof}

Next, we show that the hardness results extend to all inter-available-independent algorithms, as well as to the preemptive edge-arrival models. In the \textsf{Online Edge Arrival Matching} problem~\cite{esa/BuchbinderST17}, edges arrive one by one, and upon the arrival of each edge, the algorithm must irrevocably decide whether to accept it into the matching. In the preemptive setting, namely \textsf{Online Preemptive Matching}~\cite{stacs/EpsteinLSW13, approx/McGregor05}, the algorithm is additionally allowed to remove previously accepted edges before accepting a newly arrived edge.

The intuition behind the two corollaries is essentially the same. At the deadline of a vertex, all of its available neighbors can be made indistinguishable to the algorithm. Consequently, the adversary can defer assigning their identities until after the algorithm has already made its matching decision. Therefore, no algorithm can exploit any structural difference among these neighbors, and a balanced decision, such as the one made by \wtf, is the best that any algorithm can do.

\begin{corollary}
    \label{cor:hard_inter_independent}
    No inter-available-independent algorithm can achieve a competitive ratio better than $2-\sqrt{2}$ for the fractional fully online matching problem, even on bipartite graphs.
\end{corollary}

\begin{corollary}
    \label{cor:hard_other}
    No fractional algorithm can achieve a competitive ratio better than $2-\sqrt{2}$ for online preemptive edge arrival matching, even on bipartite graphs.
\end{corollary}

\begin{proof}[Proof of \Cref{cor:hard_other,cor:hard_inter_independent}]
We use essentially the same hard instance as before, with only a minor modification in how the input is generated and revealed. The underlying graph and the deadlines remain unchanged. Before the process begins, we independently and uniformly permute the latent identities, or equivalently the future roles, of the vertices within each relevant group, revealing these roles only when they become observable from the arriving edges.

Consider a vertex $u_{t,i}$ at its deadline. Conditional on every revealed history, its currently available neighbors whose future roles remain unrevealed are exchangeable. By symmetrizing any deterministic algorithm over the random permutations, we may therefore assume, without changing its expected performance, that it distributes the matched mass equally among these neighbors. For inter-available-independent algorithms, the required indistinguishability follows immediately from the definition of the model. Consequently, the same balanced process and the same lower-bound analysis as for \wtf apply.

For the edge-arrival model, we reveal the edges corresponding to each deadline consecutively at the same point in the arrival sequence. Since the latent roles of their endpoints remain exchangeable, the same symmetrization argument applies. Moreover, preemption provides no additional power on this instance: after the corresponding edge batch has arrived, no future edge incident to the departing endpoint will appear, and removing an already selected edge can create at most the same amount of matching through its other endpoint.

Since all random permutations are sampled before the algorithm starts, this defines an oblivious input distribution. Applying Yao's principle, we conclude that no algorithm in either model can achieve a competitive ratio better than $2-\sqrt{2}$.
\end{proof}

\section{Improved Competitive Ratio beyond Water-Filling}
In this section, we develop an algorithm that goes beyond \wtf and beats the $2-\sqrt{2}$ barrier. The key idea is \emph{eager matching}, which implicitly captures the inter-available structure of the graph and thereby bypasses the $2-\sqrt{2}$ barrier for inter-available-independent algorithms. Instead of making matching decisions only at deadlines, our algorithm also allows vertices to make matching decisions upon arrival. Such arrival-time decisions are made on inter-available edges, and thus directly exploit the inter-available structure of the graph.

Although a fixed amount of feasible matching mass can be postponed to a deadline without decreasing its immediate primal value, committing some mass upon arrival can change the state seen by later decisions.  The economic view below identifies when this effect improves the primal dual guarantee.

\paragraph{The Economic View for Primal-Dual-Based Algorithm Design.}
Within the primal-dual analytical framework, we can adopt an economic perspective to reinterpret the algorithm. Each vertex can be viewed as a selfish agent seeking to maximize its own dual variable.
Our algorithm specifies both the dual sharing rule for each fractional matching increment and the timing of each agent's matching decisions. For example, in the analysis of the \wtf algorithm, the dual-sharing function $g(x_v)$ can be regarded as a pricing function, where $g(x_v)$ represents the price of vertex $v$ when its current water level is $x_v$. Consequently, at vertex $u$'s deadline, $u$ always matches an infinitesimal amount $\diff x$ to the cheapest neighbor $v$ (which corresponds to the neighbor with the lowest water level), thereby maximizing its dual increment $(1 - g(x_v)) \diff x$, since $g$ is increasing.

Under this economic view, consider the \emph{eager} idea, where we allow a vertex to make some early fractional matching upon its arrival. Suppose vertex $u$ has a neighbor $v$ with a very low water level, and the inequality $g(x_u) \le 1 - g(x_v)$ holds. We observe that, in the worst-case analysis, it is beneficial for $u$ to match with $v$ immediately. Otherwise, $u$ might be matched later by another vertex \emph{passively} at price only $g(x_u)$, while the current dual increment rate is $1 - g(x_v) > g(x_u)$. Motivated by this observation, we propose the first improved algorithm, called \ewf, where we let every vertex $u$ make some active matching at its arrival time whenever there is a neighbor $v$ that satisfies $1-g(x_v) > g(x_u)$. The formal description and the analysis are presented in \Cref{sec:ewtf}. Furthermore, building on the eager-matching idea, we develop a finer-grained design of the price function that captures the historical information of the water levels, thereby further improving the competitive ratio in \Cref{sec:two_dim}.

\subsection{Eager Water-Filling Algorithm} \label{sec:ewtf}

In this section, we present the \ewf algorithm for fractional fully online matching and prove the following theorem.

\thmEagerWf*

% \begin{theorem}
%     \label{thm:eager-wf}
%     The \ewf algorithm can achieve a competitive ratio of $0.592$.
% \end{theorem}

\paragraph{Eager Water-Filling.}
Fix a continuous and strictly increasing function
$g:[0,1]\to[0,1]$ satisfying $g(0)=0$ and $g(1)=1$.
Initialize all variables $x_{uv}$ and $\alpha_u$ to be zero. 

\begin{enumerate}
	\item Upon the arrival of a vertex $u$, $u$ continuously matches the neighbor $v$ with the lowest water level as long as $g(x_u) + g(x_v) \leq 1$.
	In other words, the process increases $x_u$ and the lowest water level among the neighbors of $u$ until $g(x_u) + g(x_v) > 1$ for all neighbors $v$ of $u$.
	\item At the deadline of $u$, $u$ continuously matches the neighbor $v$ with the lowest water level until $x_u = 1$, or $x_v = 1$ for all neighbors $v$ of $u$.
\end{enumerate}
The pseudocode is presented in \Cref{alg:ewf}.
Note that the second step of \ewf is the same as \wtf, and the only difference is the \emph{eager} part (i.e., step (1)).

In both steps, when we match $u$ with its neighbor $v$, we consider $u$ as the active vertex and $v$ as the passive vertex.
When $x_{uv}$ increases by $\dd x$, we update the dual variables $\alpha_u$ and $\alpha_v$ as follows:
\begin{equation*}
\dd\alpha_u = (1-g(x_v))\dd x \quad \text{and} \quad \dd\alpha_v = g(x_v) \dd x.
\end{equation*}

\begin{algorithm}[ht]
\caption{Eager Water-Filling}
\label{alg:ewf}
\SetAlgoLined
\SetKwIF{When}{ElseWhen}{Else}{Trigger: when}{do}{Else when}{else}{end}
\KwData{Graph $G=(V,E)$ with vertex arrivals and departures}
\KwResult{Fractional matching $\{x_{uv}\}$ and dual variables $\{\alpha_u\}$}

\textbf{Initialization:} Initialize each variable only once: when a vertex $u$ arrives, set $x_u\gets 0$ and $\alpha_u\gets 0$; when an edge $(u,v)$ is revealed for the first time, set $x_{uv}\gets 0$.\;

\When{a vertex $u$ arrives}{
  $N(u)\gets \{\,v:\ v\text{ is a neighbor of }u\text{ and currently available}\,\}$\tcp*{available neighbors at arrival}
  $v\gets \arg\min_{w\in N(u)} g(x_w)$\;
  \While{$g(x_u)+g(x_v)\le 1$}{
    $x_{uv}\gets x_{uv}+\diff x$\tcp*{match a small amount $\diff x$ on $(u,v)$}
    $x_u\gets x_u+\diff x$, \quad $x_v\gets x_v+\diff x$\tcp*{water-level updates}
    $\alpha_u\gets \alpha_u+(1-g(x_v))\,\diff x$, \quad $\alpha_v\gets \alpha_v+g(x_v)\,\diff x$\tcp*{dual updates}
    $v\gets \arg\min_{w\in N(u)} g(x_w)$\;
  }
}

\When{a vertex $u$ departs}{
  $N(u)\gets \{\,v:\ v\text{ is a neighbor of }u\text{ and currently available}\,\}$\tcp*{available neighbors at departure}
  \While{$x_u < 1$ and $\min_{w \in N(u)} x_w < 1$}{
    $v\gets \arg\min_{w\in N(u)} g(x_w)$\;
    $x_{uv}\gets x_{uv}+\diff x$\tcp*{match a small amount $\diff x$ on $(u,v)$}
    $x_u\gets x_u+\diff x$, \quad $x_v\gets x_v+\diff x$\tcp*{water-level updates}
    $\alpha_u\gets \alpha_u+(1-g(x_v))\,\diff x$, \quad $\alpha_v\gets \alpha_v+g(x_v)\,\diff x$\tcp*{dual updates}
  }
}
\end{algorithm}

%\begin{theorem}\label{th:ewf}
%	\ewf is $0.592$-competitive for \fofmp.
%\end{theorem}

\subsection{Analysis of Eager Water-Filling}
\label{sec:ewf_analysis}

By Lemma~\ref{lem:primal-dual}, it suffices to show that for any pair of neighbors $u$ and $v$ we have $\alpha_u+\alpha_v \geq \Gamma$ in order to prove \ewf is $\Gamma$-competitive.
Fix any pair of neighbors $u$ and $v$, and assume $u$ has an earlier deadline than $v$.

Let $p_u$ be the water level of $u$ right \emph{before} $u$'s deadline.
Let $p_v$ be the water level of $v$ right \emph{after} $u$'s deadline.
Let $t_u, t_v$ be the water levels of $u$ and $v$ right \emph{after} their arrivals, respectively.
We prove the following lower bound on the combined gain of $u$ and $v$.

\begin{lemma}\label{lemma:gain-of-u-and-v-ewf}
Right after $u$'s deadline, we have
\begin{equation}
\alpha_v + \alpha_u\ge t_v \cdot g(t_v) + \int_{t_v}^{p_v} g(x) \dd{x} + t_u \cdot g(t_u) + \int_{t_u}^{p_u} g(x) \dd{x}+ (1-p_u) \cdot (1-g(p_v)).\label{eq:gain-ewf}
\end{equation}
\end{lemma}
\begin{proof}
	Recall that when $v$ arrives, $v$ matches some neighbor actively until $x_v = t_v$.
	Moreover, when $x_v$ increases (actively) from $0$ to $t_v$, the neighbor $z$ it matches always satisfies $g(x_z) + g(t_v) \leq 1$.
	Thus when $x_v$ increases by $\dd x$ the gain of $v$ is $(1-g(x_z))\dd x \geq g(t_v)\dd x$.
	Hence right after $v$'s arrival we have $\alpha_v \geq t_v\cdot g(t_v)$.
	When $x_v$ further increases from $t_v$ to $p_v$ between $v$'s arrival and $u$'s deadline, $\alpha_v$ increases at the rate of $g(x_v)$.
	Thus after $u$'s deadline we have $\alpha_v \ge t_v \cdot g(t_v) + \int_{t_v}^{p_v} g(x) \dd{x}$.
    
	Similarly, right before $u$'s deadline we have $\alpha_u \ge t_u \cdot g(t_u) + \int_{t_u}^{p_u} g(x) \dd{x}$.
	If $p_v = 1$, then $(1-p_u)\cdot(1-g(p_v)) = 0$ and the statement is proved.
	Otherwise, at $u$'s deadline, $x_u$ increases (actively) from $p_u$ to $1$, and $u$ always matches a neighbor with a water level at most $p_v$.
	Thus after the deadline of $u$ we have $\alpha_u \ge t_u \cdot g(t_u) + \int_{t_u}^{p_u} g(x) \dd{x} + (1-p_u)\cdot (1-g(p_v))$.

	Adding the lower bounds for $\alpha_u$ and $\alpha_v$ concludes the proof.
\end{proof}

\paragraph{Comparison with \wtf.}
We compare this with the competitive analysis of \wtf. Let $p_u,p_v$ be defined for \wtf in the same way as for \ewf, with the dual variables also updated in the same way. Observe that \wtf is exactly the second step of our \ewf algorithm. As shown in \Cref{eqn:wtf_gain}, we have
\begin{equation*}
\alpha_u+\alpha_v \geq \int_0^{\pw_u} g(x) \dd x + (1-{\pw_u})(1-g(p_v)) + \int_0^{p_v} g(x)\dd x.
\end{equation*}

Observe that \Cref{eq:gain-ewf} is at least as good as \Cref{eqn:wtf_gain}, because $t \cdot g(t) \geq \int_0^{t} g(x) \dd x$ holds for all $t$.
On the other hand, we have not shown any constraint on the values of $t_u,t_v$.
In the case when $t_u = t_v = 0$, \Cref{eq:gain-ewf} degenerates to \Cref{eqn:wtf_gain}.

We continue our analysis by observing that if $v$ arrives earlier than $u$, then right after $u$'s arrival, we have $g(t_u) + g(x_v) > 1$, and $x_v \leq p_v$.
Thus, we have the constraint that $g(t_u) + g(p_v) > 1$.
%Consequently, if $t_u$ is small then $p_v$ must be relatively large.
Similarly, if $u$ arrives earlier than $v$ then we have $g(t_v) + g(p_u) > 1$.

Combining the constraints on $t_u,p_u,t_v,p_v$ with Lemma~\ref{lemma:gain-of-u-and-v-ewf}, we show that there exists a function $g$ such that the total gain of $u$ and $v$ combined is strictly larger than the ratio $2-\sqrt{2}$ that is proved tight for \wtf.

\subsection{Reformulating the Lower Bound}

It remains to find a continuous and strictly increasing function $g$ such that the minimum of the RHS of \Cref{eq:gain-ewf}, over possible values of $t_u,t_v,p_u,p_v$, is maximized.
In this section, we reformulate the lower bound and eliminate $t_u$ and $t_v$ from it.

Since $g$ is strictly increasing, the RHS of \Cref{eq:gain-ewf} is increasing in both $t_u$ and $t_v$. This follows from the fact that the function $t \cdot g(t) - \int_0^t g(x) \dd x$ is monotone increasing in $t$.

Thus, the minimum is achieved when $t_u$ and $t_v$ are minimized, subject to the constraint
\begin{align*}
g(t_u) + g(p_v) & > 1 \text{ if $v$ arrives earlier than $u$, or}\\
g(t_v) + g(p_u) & > 1 \text{ if $u$ arrives earlier than $v$}.
\end{align*}

Since the two stopping constraints are strict, the corresponding minima need not be attained.  However, all the expressions involved are continuous, so their infima are obtained by approaching the boundary.  Hence, for the purpose of the lower-bound calculation, we may evaluate the limiting expressions at $g(t_u)+g(p_v)=1$ or $g(t_v)+g(p_u)=1$, respectively.

Since $g$ is a continuous and strictly increasing bijection from
$[0,1]$ to itself, let $h(\cdot)=g^{-1}(\cdot)$ be its inverse.
Then $h:[0,1]\to[0,1]$ is also continuous and strictly increasing,
with $h(0)=0$ and $h(1)=1$.
For any $p\in [0,1]$, we have
\begin{equation*}
\int_0^p g(x) \dd x = p\cdot g(p) - \int_0^{g(p)} h(y) \dd y.
\end{equation*}

\begin{lemma}\label{lemma:v-arrive-earlier}
	If $v$ arrives earlier than $u$, then we have
	\begin{equation*}
	\alpha_u + \alpha_v \ge \min_q\left\{ q \cdot h(q) - \int_0^q h(y) \dd{y} + 1 - q \right\}.
	\end{equation*}
\end{lemma}
\begin{proof}
	Let $q_u=g(p_u), q_v=g(p_v)$.
    If $v$ arrives earlier than $u$, the preceding limiting argument allows us to evaluate the boundary at $t_u=g^{-1}(1-g(p_v))=h(1-q_v)$ and $t_v=0$.
	\begin{equation*}
	\alpha_v + \alpha_u\ge \int_{0}^{p_v} g(x) \dd{x} + t_u \cdot g(t_u) + \int_{t_u}^{p_u} g(x) \dd{x}+ (1-p_u) \cdot (1-g(p_v)).
	\end{equation*}

	 For $p'_u\geq p_u\geq t_u$, the change in the relevant expression is
    \[
        \int_{p_u}^{p'_u}\bigl(g(x)+g(p_v)-1\bigr)\dd x\geq 0,
    \] 
    because $g(x)\geq g(t_u)=1-g(p_v)$.  Hence the expression is nondecreasing in $p_u$, and its infimum is attained at $p_u=t_u$. Thus we have
	\begin{equation*}
	\alpha_v + \alpha_u\ge \int_{0}^{p_v} g(x) \dd{x} + t_u \cdot g(t_u) + (1-t_u) \cdot (1-g(p_v)).
	\end{equation*}

	Using $t_u = h(1-q_v)$ we have $g(t_u) = 1-q_v = 1-g(p_v)$, which implies
	\begin{equation*}
	\alpha_v + \alpha_u\ge \int_{0}^{p_v} g(x) \dd{x} + 1-g(p_v) = h(q_v)\cdot q_v - \int_0^{q_v} h(y) \dd y + 1 - q_v.
	\end{equation*}

	Taking the minimum of the RHS over $q_v$ yields the lemma.
\end{proof}

\begin{lemma}\label{lemma:u-arrive-earlier}
	If $u$ arrives earlier than $v$, then we have
	\begin{equation*}
	\alpha_u + \alpha_v \ge
	\min_{q_u,q_v} \Big\{ q_u\cdot h(q_u) - \int_0^{q_u}h(y) \dd y
	+ q_v\cdot h(q_v) - \int_0^{q_v}h(y) \dd y
	+ \int_0^{1-q_u} h(y) \dd y + (1-h(q_u)) \cdot (1-q_v) \Big\}.
	\end{equation*}
\end{lemma}
\begin{proof}
	Let $q_u=g(p_u), q_v=g(p_v)$.
    If $u$ arrives earlier than $v$, we evaluate the boundary at $t_v=g^{-1}(1-g(p_u))=h(1-q_u)$ and $t_u=0$.
	\begin{align*}
	\alpha_v + \alpha_u
    & \ge t_v \cdot g(t_v) + \int_{t_v}^{p_v} g(x) \dd{x} + \int_{0}^{p_u} g(x) \dd{x} + (1-p_u) \cdot (1-g(p_v)) \\
	& =  t_v\cdot g(t_v) + \int_{0}^{p_v} g(x) \dd{x} - \int_0^{t_v} g(x)\dd x + \int_{0}^{p_u} g(x) \dd{x} + (1-p_u) \cdot (1-g(p_v)) \\
	& =  h(1-q_u)\cdot (1-q_u) + \left( h(q_v)\cdot q_v - \int_0^{q_v}h(y)\dd y \right) - \left( h(1-q_u)\cdot (1-q_u) - \int_0^{1-q_u} h(y) \dd y \right) \\
	& \quad + \left( h(q_u)\cdot q_u - \int_0^{q_u}h(y)\dd y \right) + (1-h(q_u))\cdot (1-q_v) \\
	& =  q_u\cdot h(q_u) - \int_0^{q_u}h(y) \dd y
	+ q_v\cdot h(q_v) - \int_0^{q_v}h(y) \dd y
	+ \int_0^{1-q_u} h(y) \dd y + (1-h(q_u)) \cdot (1-q_v).
	\end{align*}

	Taking the minimum of the RHS over $q_u$ and $q_v$ yields the lemma.
\end{proof}

Finally, we use factor-revealing LP techniques to find a function $h$ with the following property. The proof of Lemma~\ref{lemma:lower-bound-ratio-ewf} is deferred to Appendix~\ref{sec:factor_revealing_lp_ewf}.

\begin{lemma}\label{lemma:lower-bound-ratio-ewf}
    There exists a continuous and strictly increasing function
    $h:[0,1]\to[0,1]$ satisfying the displayed inequalities and endpoint
    conditions for $\Gamma=0.592$.
	\begin{align}
	\forall q\in[0,1], \qquad & q \cdot h(q) - \int_0^q h(y) \dd{y} + 1 - q \geq \Gamma, \label{eqn:v_earlier}\\
	\forall q_u,q_v\in[0,1],\qquad & q_u\cdot h(q_u) - \int_0^{q_u}h(y) \dd y
	+ q_v\cdot h(q_v) - \int_0^{q_v}h(y) \dd y \notag \\
	& \qquad + \int_0^{1-q_u} h(y) \dd y + (1-h(q_u)) \cdot (1-q_v) \geq \Gamma, \label{eqn:u_earlier}\\
	& h(0) = 0, h(1) = 1.
	\end{align}
\end{lemma}

\begin{theorem}[The Formal Version of Theorem~\ref{thm:eager-wf}]
	Eager Water-Filling with the function $g=h^{-1}$, where $h$ is chosen in Lemma~\ref{lemma:lower-bound-ratio-ewf}, is $0.592$-competitive for fractional fully online matching on general graphs.
\end{theorem}
\begin{proof}
	We conclude the competitive ratio of \ewf by putting the lemmas together. Approximate dual feasibility follows by the two cases.
	If $v$ arrives earlier than $u$, we have
	\begin{align*}
	\alpha_u+\alpha_v
    & \ge \min_q\left\{ q \cdot h(q) - \int_0^q h(y) \dd{y} + 1 - q \right\} \tag{Lemma~\ref{lemma:v-arrive-earlier}} \\
	& \ge \Gamma = 0.592~. \tag{\Cref{eqn:v_earlier}}
	\end{align*}
	If $u$ arrives earlier than $v$, we have
	\begin{align*}
	\alpha_u + \alpha_v
    & \ge  \min_{q_u,q_v} \bigg\{ q_u\cdot h(q_u) - \int_0^{q_u}h(y) \dd y + q_v\cdot h(q_v) - \int_0^{q_v}h(y) \dd y \\
	& \phantom{\min_{q_u,q_v} \bigg\{} + \int_0^{1-q_u} h(y) \dd y + (1-h(q_u)) \cdot (1-q_v) \bigg\} \tag{Lemma~\ref{lemma:u-arrive-earlier}} \\
	& \ge  \Gamma = 0.592~. \tag{\Cref{eqn:u_earlier}}
	\end{align*}
	Finally, we conclude the theorem by \Cref{lem:primal-dual}.
\end{proof}

\subsection{Eager Water-Filling with History-Based Pricing}
\label{sec:two_dim}

In this section, we refine the \ewf algorithm by introducing a finer-grained pricing mechanism. In both \wtf and \ewf, a vertex's water level can be accumulated through two distinct channels:
\begin{enumerate}
    \item \textbf{Active channel:} the increase resulting from the vertex actively matching its neighbors at its own arrival time (i.e., during the eager matching phase) or at its deadline.
    \item \textbf{Passive channel:} the increase resulting from the vertex being matched by its neighbors at their respective arrival times or deadlines.
\end{enumerate}

In the \wtf algorithm, when a vertex makes its matching decision at its deadline, the water levels of all its available neighbors must have been accumulated only through the passive channel. This is because once a vertex receives an active water level, it reaches its deadline and is no longer available. In contrast, under the \ewf algorithm, when a vertex makes a decision, either upon arrival or at its deadline, some of its available neighbors may already have active water levels accumulated at their own arrival times. As a result, two available neighbors with the same total water level may contribute differently to the dual objective, depending on how their water levels were accumulated.

Building on this distinction, we introduce a history-based pricing function defined over two dimensions: $g(t_v, x_v)$. Here, $t_v$ denotes the active water level accumulated upon $v$'s arrival, while $x_v$ represents the current total water level (including both active and passive portions). Intuitively, this pricing scheme enables the algorithm to make more balanced matching decisions by leveraging the historical information of how the water level was obtained.

To illustrate the advantage of such history-based pricing, consider what happens under the one-dimensional pricing scheme before the history information is introduced. Suppose an arriving vertex $u$ has two neighbors, $v_1$ and $v_2$, sharing the same total water level (i.e., $x_{v_1} = x_{v_2} = x$) but with different histories. Assume that $v_1$ has a high active water level $t_{v_1}$ (obtained upon its own arrival), while $v_2$ is entirely passive (i.e., $t_{v_2} = 0$). Since the one-dimensional price $g(x)$ depends only on the current water level, these two neighbors are indistinguishable to the algorithm; consequently, $u$ matches them at an equal rate, causing $\alpha_{v_1}$ and $\alpha_{v_2}$ to increase at the same rate. However, $v_1$ has typically accumulated significantly more dual value than $v_2$. Indeed, the active increments that created the active water level of $v_1$ were generated under the one-dimensional eager condition $1-g(x_{\text{neighbor}}) \ge g(x_{v_1})$ (equivalently, $g(x_{v_1})+g(x_{\text{neighbor}})\le 1$), and thus their marginal dual gain is generally higher than that of passive increments. Since our competitive analysis focuses on the worst-case pair, treating
vertices with the same total water level but different matching histories
identically may be suboptimal.  This motivates a two-variable pricing
function $g(t,x)$ that can distinguish vertices with the same current water
level but different matching histories.

We formalize the idea as follows.

% In this section, following the economic view and the eager idea,
% we further extend the purely ``current water level'' based pricing rule
% to a history-based approach with a two-dimensional pricing function.
% Moreover, we introduce a more flexible policy to determine the stopping time of the eager matching at each vertex's arrival.

\paragraph{Eager Water-Filling with History-Based Pricing.}
Fix a continuous and strictly increasing function
$f:[0,1]\to[0,1]$ satisfying $f(0)=0$ and $f(1)=1$, and a continuous
function $g(t,x)$, defined for $0\leq t\leq x\leq1$, such that, for
every fixed $t$, the function $x\mapsto g(t,x)$ is strictly increasing
on $[t,1]$.  We further require $g(t,t)=f(t)$ and $g(t,1)=1$ for every $t\in[0,1]$.
% \paragraph{Eager Water-Filling with History-Based Pricing.} Fix an increasing function $f:[0,1] \to [0,1]$ and a two-dimensional function $g:[0,1]^2 \to [0,1]$ that is non-decreasing in both dimensions and satisfies $g(t,1) = 1$ for all $t\in [0,1]$.
%\yuhao{is that non-decreasing?}
%\zhihao{do we need the monotonicity in the first dimension?}
All matched portions $x_{uv}$, water levels $x_u$, and dual variables $\alpha_u$ are initialized to $0$ upon $u$'s arrival.
%We use $\rho_v = g(t_v, x_v)$ to denote the current price of vertex $v$.
\begin{itemize}
	%\item Each vertex $v$ set its price as $g(t_v, x_v)$, where $t_v$ is the active water level of $v$ after its arrival and $x_v$ is the current water level of $v$.
	\item Upon the arrival of each vertex $u$, fractionally match $u$ to the neighbor $v$ with the cheapest price $\rho_v = g(t_v, x_v)$, as long as $f(x_u) + g(t_v, x_v) \le 1$. When the process ends, let $t_u = x_u$ be the \emph{active} water level of $u$.% be By the definition of our algorithm, we have $g(\alpha_u, x_u) \ge \alpha_u$ right after the arrival of $u$. If there is no available neighbor of $u$, let $\alpha_u$ be $0$.
	\item At the deadline of each vertex $u$, fractionally match it to the neighbor $v$ with the cheapest price $g(t_v, x_v)$.
\end{itemize}
In both steps, when $x_{uv}$ increases by $\diff x$ at the arrival or deadline of $u$, update $\alpha_u, \alpha_v$ by
\[
\diff\alpha_u = (1-g(t_v,x_v))\,\diff x \quad \text{and} \quad \diff\alpha_v = g(t_v, x_v)\,\diff x
\]
%\zhihao{the pseudocode does not look good. I prefer not to include it.}
See Algorithm~\ref{alg:fully} for the pseudocode of our algorithm. We remark that \ewf is a special case of the above algorithm when $g$ is set to the same function $f$, i.e., $g(t, x) = f(x)$ for all $t$.

% Changed to Algorithm2e format since the EC format.

\begin{algorithm}[t]
\caption{Eager Water-Filling with History-Based Pricing}
\label{alg:fully}
\SetAlgoLined
\SetKwIF{When}{ElseWhen}{Else}{Trigger: when}{do}{Else when}{else}{end}
\KwData{Graph $G=(V,E)$ with vertex arrivals and departures}
\KwResult{Fractional matching $\{x_{uv}\}$ and dual variables $\{\alpha_u\}$}

\textbf{Initialization:} Initialize each variable only once: when a vertex $u$ arrives, set $x_u\gets 0$ and $\alpha_u\gets 0$; when an edge $(u,v)$ is revealed for the first time, set $x_{uv}\gets 0$.\;

\When{a vertex $u$ arrives}{
  $N(u)\gets \{\,v:\ v\text{ is a neighbor of }u\text{ and currently available}\,\}$\tcp*{available neighbors at arrival}
  $v\gets \arg\min_{w\in N(u)} g(t_w,x_w)$\;
  \While{$f(x_u)+g(t_v,x_v)\le 1$}{
    $x_{uv}\gets x_{uv}+\diff x$\tcp*{match a small amount $\diff x$ on $(u,v)$}
    $x_u\gets x_u+\diff x$, \quad $x_v\gets x_v+\diff x$\tcp*{water-level updates}
    $\alpha_u\gets \alpha_u+(1-g(t_v,x_v))\,\diff x$, \quad $\alpha_v\gets \alpha_v+g(t_v,x_v)\,\diff x$\tcp*{dual updates}
    $v\gets \arg\min_{w\in N(u)} g(t_w,x_w)$\;
  }
  $t_u\gets x_u$\tcp*{fix $u$'s active water level}
}

\When{a vertex $u$ departs}{
  $N(u)\gets \{\,v:\ v\text{ is a neighbor of }u\text{ and currently available}\,\}$\tcp*{available neighbors at departure}
  \While{$x_u < 1$ and $\min_{w \in N(u)} x_w < 1$}{
    $v\gets \arg\min_{w\in N(u)} g(t_w,x_w)$\;
    $x_{uv}\gets x_{uv}+\diff x$\tcp*{match a small amount $\diff x$ on $(u,v)$}
    $x_u\gets x_u+\diff x$, \quad $x_v\gets x_v+\diff x$\tcp*{water-level updates}
    $\alpha_u\gets \alpha_u+(1-g(t_v,x_v))\,\diff x$, \quad $\alpha_v\gets \alpha_v+g(t_v,x_v)\,\diff x$\tcp*{dual updates}
  }
}
\end{algorithm}

\subsection{Analysis of Eager Water-Filling with History-Based Pricing}
In this section, we analyze the competitive ratio of \Cref{alg:fully}. Our proofs closely follow the analysis of \ewf in \Cref{sec:ewf_analysis}. Each lemma presented below corresponds to a counterpart in that analysis, as the family of algorithms we propose includes \ewf as a special case.
However, we still need to perform a non-trivial transformation of the derived lower bounds to make them compatible with the factor-revealing LP framework used to optimize the functions $f$ and $g$ in our algorithm.

Fix a pair of neighbors $(u,v)$. Suppose $u$'s deadline is before $v$'s deadline. Let $p_u$ be the water level of $u$ right \emph{before} $u$'s deadline.
%We shall have that $g(\alpha_u, t_u) \ge \alpha_u$.
Let $p_v$ be the water level of $v$ right \emph{after} $u$'s deadline.

\begin{lemma}
Right after $u$'s deadline, we have
\begin{equation}
\label{eq:gain}
\alpha_u + \alpha_v \ge t_u \cdot f(t_u) + \int_{t_u}^{p_u} g(t_u, x_u) dx_u + t_v \cdot f(t_v) + \int_{t_v}^{p_v} g(t_v, x_v) dx_v + (1-p_u) \cdot (1-g(t_v, p_v)).
\end{equation}
\end{lemma}

\begin{proof}
Upon the arrival of vertex $v$, it fractionally matches the neighbor $z$ with the lowest price as long as $f(x_v) + g(t_z, x_z) \le 1$. Hence, during the process, the marginal gain of $v$ is $\diff \alpha_v = (1-g(t_z, x_z))\,\diff x \ge f(t_v)\,\diff x$. Here, the inequality holds since the utility $1-g(t_z,x_z)$ decreases during the process as the cheapest price increases. We use the fact that $f$ is a continuous function.

Thus, $\alpha_v \ge t_v \cdot f(t_v)$ after $v$'s arrival. As $v$ is later matched and $x_v$ increases from $t_v$ to $p_v$, $\alpha_v$ increases at the rate of $g(t_v, x_v)$. Hence, $\alpha_v \ge t_v \cdot f(t_v) + \int_{t_v}^{p_v}g(t_v, x_v)dx_v$ right after $u$'s deadline.

Similarly, $\alpha_u \ge t_u \cdot f(t_u) + \int_{t_u}^{p_u}g(t_u, x_u)dx_u$ right \emph{before} $u$'s deadline.
At the deadline of $u$, $u$ actively matches the neighbor with the lowest price, which is at most $g(t_v, p_v)$. Therefore, after the deadline of $u$, $\alpha_u \ge t_u \cdot f(t_u) + \int_{t_u}^{p_u}g(t_u, x_u)dx_u + (1-p_u) \cdot (1-g(t_v, p_v))$.

Summing up the total gains of $\alpha_u, \alpha_v$ gives the claimed bound.
\end{proof}

\subsection{Reformulating the Lower Bound}
For $0\leq\tau\leq\theta\leq1$, let $h(\tau,\theta)$ be the unique
$x\in[f^{-1}(\tau),1]$ such that
$
g(f^{-1}(\tau),x)=\theta.
$
Then $h(\tau,\tau)=f^{-1}(\tau)$ and $h(\tau,1)=1$.

% To facilitate the analysis, we define the following function $h: [0,1]^2 \to [0,1]$ that is in one-to-one correspondence with $g$:
% \[
% h(\tau,\theta) = \inf \{x: g(f^{-1}(\tau), x) \ge \theta\}.
% \]
% Our function $g$ is defined only on $(t, x) \in [0,1]^2$ with $x \ge t$, and $f(x) = g(x,x)$ for all $x \in [0,1]$.
% We remark that $h$ is non-decreasing in the second dimension, and $h(\tau,\tau)=f^{-1}(\tau)$.

\begin{lemma} \label{lemma:int_g_to_int_h_2D}
    For any $0 \le \tau \le \theta \le 1$, we have
    %\zhihao{Seems we need that $g(\alpha, h(\alpha, \alpha)) = \alpha$}
    \[
    \int_{h(\tau,\tau)}^{h(\tau,\theta)} g(f^{-1}(\tau),x) dx = \theta \cdot h(\tau, \theta) - \tau \cdot h(\tau, \tau) - \int_{\tau}^{\theta} h(\tau, y) dy.
    \]
\end{lemma}
\begin{proof}
Fix $\tau$, and let
\[
a=h(\tau,\tau)=f^{-1}(\tau),\qquad b=h(\tau,\theta).
\]
For every $y\in[\tau,\theta]$, by the definition of $h$ and the monotonicity of
$g$ in the second coordinate,
\[
\bigl|\{x\in[a,b]: g(a,x)\ge y\}\bigr|=b-h(\tau,y).
\]
Thus, by applying Fubini’s theorem,
\[
\begin{aligned}
\int_a^b g(a,x)\,dx
&=\tau(b-a)+\int_\tau^\theta
\bigl|\{x\in[a,b]: g(a,x)\ge y\}\bigr|\,dy \\
&=\tau(b-a)+\int_\tau^\theta (b-h(\tau,y))\,dy \\
&=\theta b-\tau a-\int_\tau^\theta h(\tau,y)\,dy .
\end{aligned}
\]
Substituting $a=h(\tau,\tau)$ and $b=h(\tau,\theta)$ yields the lemma.
\end{proof}

The following lemma studies the two cases depending on whether $u$ arrives earlier than $v$ or $v$ arrives earlier than $u$ and reformulates the lower bound~\eqref{eq:gain}.
We first introduce some notation.

Let $\tau_u = f(t_u), \tau_v = f(t_v), \theta_u = g(t_u, p_u)$ and $\theta_v = g(t_v, p_v)$.
Note that since $\tau_u = g(t_u, t_u)$ and $p_u \geq t_u$, we have $\theta_u \geq \tau_u$.
Similarly, we have $\theta_v \geq \tau_v$.
Furthermore, we have
\begin{equation*}
    h(\tau_u, \tau_u) = t_u, \quad
    h(\tau_u, \theta_u) \leq p_u, \quad
    h(\tau_v, \tau_v) = t_v, \quad \text{ and } \quad
    h(\tau_v, \theta_v) \leq p_v.
\end{equation*}

\begin{lemma}
\label{lem:uearlier}
If $u$ arrives earlier than $v$, then
\begin{align}
    \alpha_u + \alpha_v \ge \min_{\substack{\tau_u \le \theta_u \\
    1-\theta_u \le \tau_v \le \theta_v}} \bigg\{ \theta_u \cdot h(\tau_u, \theta_u) & - \int_{\tau_u}^{\theta_u} h(\tau_u, y) dy \nonumber \\
    & + \theta_v \cdot h(\tau_v, \theta_v) - \int_{\tau_v}^{\theta_v} h(\tau_v, y) dy + (1-h(\tau_u, \theta_u)) \cdot (1-\theta_v) \bigg\}.
\label{eq:uearlier}
\end{align}
% \begin{multline}
% \label{eq:uearlier}
% \alpha_u + \alpha_v \ge \min_{\substack{\tau_u \le \theta_u \\ 1-\theta_u \le \tau_v \le \theta_v}} \bigg\{ \theta_u \cdot h(\tau_u, \theta_u) - \int_{\tau_u}^{\theta_u} h(\tau_u, y) dy \\
% + \theta_v \cdot h(\tau_v, \theta_v) - \int_{\tau_v}^{\theta_v} h(\tau_v, y) dy + (1-h(\tau_u, \theta_u)) \cdot (1-\theta_v) \bigg\}
% \end{multline}
\end{lemma}
\begin{proof}
% Let $\tau_u = f(t_u), \tau_v= f(t_v), \theta_u = g(t_u, p_u)$ and $\theta_v = g(t_v, p_v)$.
%By definition, we have $p_u \ge h(\alpha_u,\rho_u)$ and $p_v \ge h(\alpha_v, \rho_v)$.
With the notion of $h$, we rewrite Equation~\eqref{eq:gain}.
Notice that
\begin{align}
    t_v \cdot f(t_v) + \int_{t_v}^{p_v} g(t_v, x_v) dx_v & = \tau_v \cdot h(\tau_v,\tau_v) + \int_{h(\tau_v,\tau_v)}^{p_v} g(f^{-1}(\tau_v), x_v) dx_v \nonumber \\
    & \ge \tau_v \cdot h(\tau_v,\tau_v) + \int_{h(\tau_v,\tau_v)}^{h(\tau_v, \theta_v)} g(f^{-1}(\tau_v), x_v) dx_v \nonumber \\
    & = \theta_v \cdot h(\tau_v, \theta_v) - \int_{\tau_v}^{\theta_v} h(\tau_v,y) dy,
    \label{eq:uearlier_gainv}
\end{align}
% \begin{multline}
% \label{eq:uearlier_gainv}
% t_v \cdot f(t_v) + \int_{t_v}^{p_v} g(t_v, x_v) dx_v = \tau_v \cdot h(\tau_v,\tau_v) + \int_{h(\tau_v,\tau_v)}^{p_v} g(f^{-1}(\tau_v), x_v) dx_v \\
% \ge \tau_v \cdot h(\tau_v,\tau_v) + \int_{h(\tau_v,\tau_v)}^{h(\tau_v, \theta_v)} g(f^{-1}(\tau_v), x_v) dx_v = \theta_v \cdot h(\tau_v, \theta_v) - \int_{\tau_v}^{\theta_v} h(\tau_v,y_v) dy_v,
% \end{multline}
where the inequality follows from the fact that $p_v \ge h(\tau_v,\theta_v)$ and the last equality follows from Lemma~\ref{lemma:int_g_to_int_h_2D}.

By the monotonicity of $g$ in its second argument, the RHS of
Equation~\eqref{eq:gain} is non-decreasing in $p_u$. We obtain:
\[
g(t_u, p_u) - 1 + g(t_v, p_v) \ge g(t_u, p_u) - 1 + f(t_v) \ge 0,
\]
where the last inequality holds due to the stopping condition of our algorithm on $v$'s arrival.
Consequently, the RHS of Equation~\eqref{eq:gain} achieves its minimum when $p_u = h(\tau_u,\theta_u)$.
Therefore, we have
\begin{align}
&\phantom{\ge} t_u \cdot f(t_u) + \int_{t_u}^{p_u} g(t_u, x_u) dx_u + (1-p_u) \cdot (1-g(t_v,p_v)) \notag \\
& \ge \tau_u \cdot h(\tau_u,\tau_u) + \int_{h(\tau_u,\tau_u)}^{h(\tau_u,\theta_u)} g(f^{-1}(\tau_u),x_u) dx_u + (1-h(\tau_u,\theta_u)) \cdot (1-\theta_v) \notag \\
& = \theta_u \cdot h(\tau_u, \theta_u) - \int_{\tau_u}^{\theta_u} h(\tau_u, y) dy + (1-h(\tau_u, \theta_u)) \cdot (1-\theta_v), \label{eq:uearlier_gainu}
\end{align}
where the last equality follows by Lemma~\ref{lemma:int_g_to_int_h_2D}.
We conclude the proof by summing up \eqref{eq:uearlier_gainv} and \eqref{eq:uearlier_gainu} and taking the minimum over all possible values of $\tau_u,\theta_u,\tau_v$ and $\theta_v$. Note that since $u$ arrives earlier than $v$, we must have $\tau_v = f(t_v) \geq 1-g(t_u,p_u) = 1-\theta_u$, by the design of our algorithm.
\end{proof}

\begin{lemma}
\label{lem:vearlier}
If $v$ arrives earlier than $u$, then
\begin{equation}
\label{eq:vearlier}
	\alpha_u + \alpha_v \ge \min_{\substack{\tau_v \le \theta_v}} \left\{ \theta_v \cdot h(\tau_v, \theta_v) - \int_{\tau_v}^{\theta_v} h(\tau_v, y) dy + 1 - \theta_v \right\}.
\end{equation}
\end{lemma}
\begin{proof}
% Let $\tau_u = f(t_u), \tau_v = f(t_v), \theta_u = g(t_u,p_u)$ and $\theta_v = g(t_v,p_v)$. Let $G$ be the RHS of Equation~\eqref{eq:gain}.
For any $p'_u\geq p_u$, the change in the relevant part of the
right-hand side of Equation~\eqref{eq:gain} is
\[
\int_{p_u}^{p'_u}
\bigl(g(t_u,x)-1+g(t_v,p_v)\bigr)\dd x\geq0.
\]
Indeed, $g(t_u,x)\geq g(t_u,t_u)=f(t_u)$, and the stopping condition
at $u$'s arrival gives
$f(t_u)-1+g(t_v,p_v)\geq0$.  Hence the expression is non-decreasing
in $p_u$, and its minimum is attained at $p_u=t_u$.  This argument
does not require differentiability of $g$. Therefore, we have
\begin{align*}
\alpha_u+\alpha_v & \ge t_v \cdot f(t_v) + \int_{t_v}^{p_v} g(t_v, x_v) dx_v + t_u \cdot f(t_u) + (1-t_u) \cdot (1-g(t_v, p_v)) \\
& \ge \tau_v \cdot h(\tau_v, \tau_v) + \int_{h(\tau_v,\tau_v)}^{h(\tau_v, \theta_v)} g(f^{-1}(\tau_v), x_v) dx_v +(1-\theta_v) \\
& = \theta_v \cdot h(\tau_v, \theta_v) - \int_{\tau_v}^{\theta_v} h(\tau_v, y) dy + 1-\theta_v,
\end{align*}
where the second inequality holds by the fact that $f(t_u) \ge 1-g(t_v,p_v) = 1-\theta_v$ and the last equality follows by Lemma~\ref{lemma:int_g_to_int_h_2D}.
Finally, taking the minimum over all possible values of $\tau_v,\theta_v$ concludes the proof.
\end{proof}

Let $\Phi_1, \Phi_2$ denote the two lower bounds derived from the above lemmas. Formally,
\begin{align*}
    &\Phi_1(\tau_v,\theta_v) \eqdef \theta_v \cdot h(\tau_v, \theta_v) - \int_{\tau_v}^{\theta_v} h(\tau_v, y) dy + 1-\theta_v; \\
    &\Phi_2(\tau_u,\theta_u,\tau_v,\theta_v) \eqdef \left(
    \begin{aligned}
        &\theta_u \cdot h(\tau_u, \theta_u) - \int_{\tau_u}^{\theta_u} h(\tau_u, y) dy + \theta_v \cdot h(\tau_v, \theta_v)\\
        & - \int_{\tau_v}^{\theta_v} h(\tau_v, y) dy + (1-h(\tau_u, \theta_u)) \cdot (1-\theta_v)
    \end{aligned}
    \right).
\end{align*}
Finally, we use the folklore factor-revealing LP techniques to optimize the function $h$. A detailed implementation of the factor-revealing LP is provided in Appendix~\ref{app:factor_lp}.
\begin{lemma}
\label{lem:opth}
There exists a continuous function $h:[0,1]^2\to[0,1]$ such that, for
every $\tau\in[0,1]$, the function
$\theta\mapsto h(\tau,\theta)$ is strictly increasing,
$h(\tau,1)=1$, and the diagonal function
$\tau\mapsto h(\tau,\tau)$ is a strictly increasing bijection from
$[0,1]$ to itself.  Moreover, for $\Gamma=0.599$,
\begin{align}
	&\Phi_1(\tau_v,\theta_v)\geq\Gamma 
    &&\text{for all }0\leq\tau_v\leq\theta_v\leq1,
    \label{eqn:h-feasible-1} \\
    &\Phi_2(\tau_u,\theta_u,\tau_v,\theta_v)\geq\Gamma
    &&\text{for all }0\leq\tau_u\leq\theta_u\leq1
    \text{ and }1-\theta_u\leq\tau_v\leq\theta_v\leq1.
    \label{eqn:h-feasible-2}
\end{align}
\end{lemma}

We finally explain how to recover the pricing functions from $h$.
Define $f$ by
\[
f^{-1}(\tau)=h(\tau,\tau).
\]
For every $0\leq t\leq x\leq1$, define $g(t,x)$ as the unique
$\theta\in[f(t),1]$ satisfying
\[
h(f(t),\theta)=x.
\]
The properties of $h$ ensure that $f$ and $g$ are well defined and
continuous.  Moreover, $g(t,t)=f(t)$, $g(t,1)=1$, and $g(t,x)$ is
strictly increasing in $x$.  Thus, the construction recovers a valid
pair of pricing functions for the algorithm.

\begin{lemma}[Formal Version of Theorem~\ref{thm:eager-history}]
	The algorithm with the functions $f,g$ constructed from $h$ as above is $\Gamma=0.599$-competitive for fractional fully online matching on general graphs.
\end{lemma}
\begin{proof}
We conclude the competitive ratio of our algorithm by putting the lemmas together. Fix an arbitrary pair of neighbors $(u,v)$ with $u$'s deadline before $v$'s deadline. Approximate dual feasibility follows by the two cases. If $u$ arrives earlier than $v$,
\[
\alpha_u + \alpha_v \ge \min_{\substack{0\le\tau_u \le \theta_u\le 1 \\ 1-\theta_u \le \tau_v \le \theta_v}} \Phi_2(\tau_u,\theta_u,\tau_v,\theta_v) \ge 0.599,
\]
where the first inequality follows from Lemma~\ref{lem:uearlier} and the second inequality follows from Lemma~\ref{lem:opth}.
If $v$ arrives earlier than $u$,
\[
\alpha_u + \alpha_v \ge \min_{0\le\tau_v\le\theta_v\le 1} \Phi_1(\tau_v,\theta_v) \ge 0.599,
\]
where the first inequality follows from Lemma~\ref{lem:vearlier} and the second inequality follows from Lemma~\ref{lem:opth}.
Each primal increment is split between its two endpoints in the dual
update, so the Equal Update property holds.  Together with the
approximate dual feasibility proved above, Lemma~\ref{lem:primal-dual} completes the proof.
\end{proof}

\section{Improved Hardness for Fractional Fully Online Matching}
\label{sec:fully_upper}

This section shows an improved upper bound for the fractional fully online matching problem. Previously, the best upper bound was $0.6297$ by \citet{orl/EcklKLS21}.

\fullyUpper*

{
% \color{red}
\paragraph{Hard Instance.} Our construction includes $(2\ell+1)n$ vertices divided into $\ell$ identical groups and a terminal buffer. Each group consists of $4$ parts, denoted by $A_k, B_k, C_k$, and $D_k$, with $\alpha\cdot n$, $\alpha \cdot n$, $(1-\alpha) \cdot n$, and $(1-\alpha) \cdot n$ vertices respectively, where $\alpha$ is a constant to be fixed later. The terminal buffer consists of two additional sets $B_{\ell+1}$ and $C_{\ell+1}$, containing $\alpha n$ and $(1-\alpha)n$ vertices, respectively.
}

% \paragraph{Hard Instance.} Our construction includes $2n \cdot \ell$ vertices divided into $\ell$ identical groups. Each group consists of $4$ parts, denoted by $A_k, B_k, C_k$, and $D_k$, with $\alpha\cdot n$, $\alpha \cdot n$, $(1-\alpha) \cdot n$, and $(1-\alpha) \cdot n$ vertices respectively, where $\alpha$ is a constant to be fixed later.
% 
We first describe the structure of the underlying graph and then define the arrivals and departures of the vertices. (Refer to Figure~\ref{fig:upper}.)

\begin{itemize}
    \item There is an upper triangle between $A_k$ and $B_k$. Let $A_k = \{a_{ki}\}_{i}$ and $B_k = \{b_{ki}\}_{i}$. The $i$-th vertex $a_{ki}$ in $A_k$ is connected to every $b_{kj}$ for $j\geq i$.
    \item There is a complete bipartite graph between $A_k$ and $C_k$.
    \item There is a complete bipartite graph between $C_k$ and $D_k \cup B_{k+1} \cup C_{k+1}$.
\end{itemize}

%\begin{itemize}
%    \item $A_k$ is connected to $B_k$ and $C_k$ as an upper triangle. That means, the $i$-th vertex in $A_k$ (i.e., $a_{ki}$) is connected to all $j$-th vertices in $B_k$ (i.e., $b_{kj}$) such that $j\geq i$, and it is connected to all vertices in $C_k$.
%    \item $C_k$ is completed connected to $D_k$, $B_{k+1}$ and $C_{k+1}$. That means, each vertex in $A_k$ is adjacent to every vertex in $D_k$, $B_{k+1}$ and $C_{k+1}$.
%\end{itemize}
\begin{figure}[ht]
    \center
    \includegraphics[scale = 0.4]{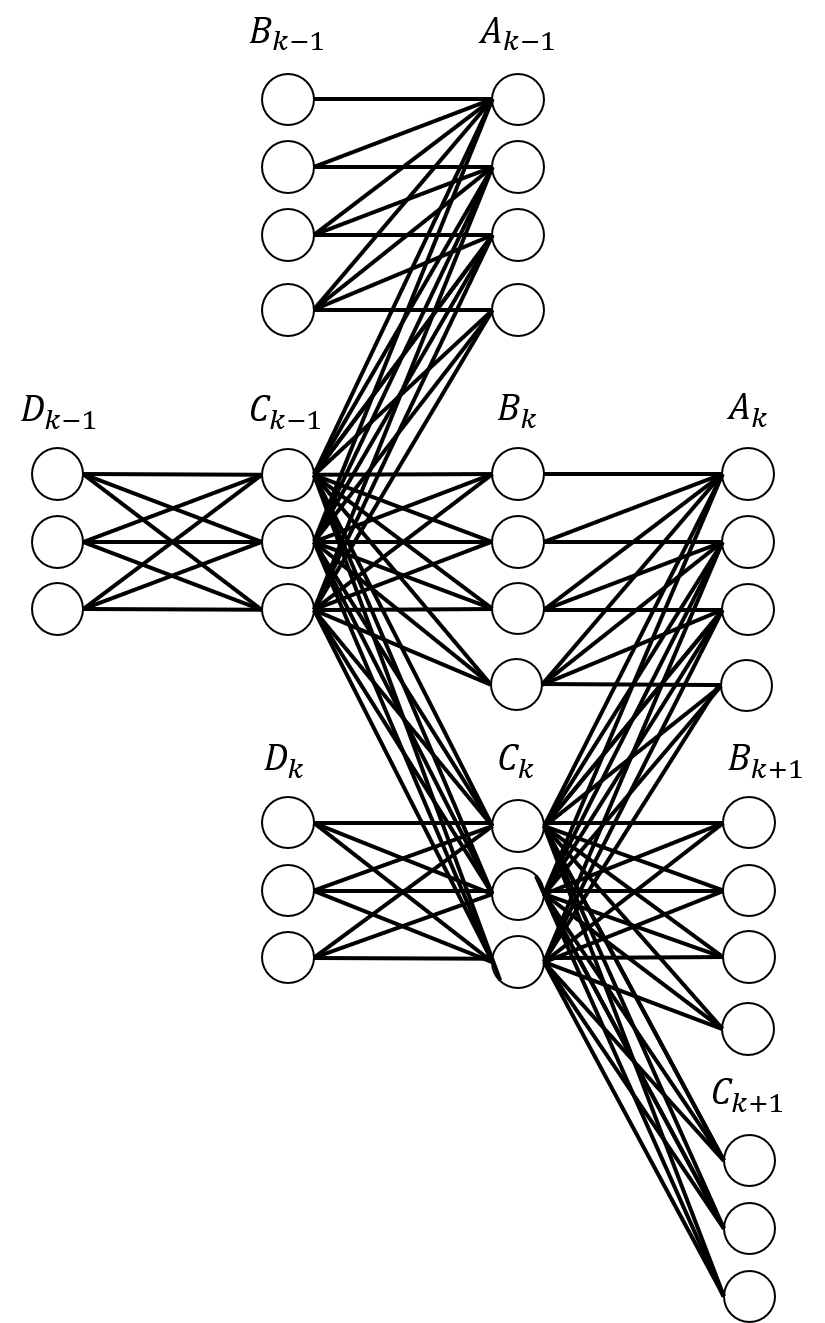}
    \caption{An illustration when $\alpha \cdot n =4$ and $(1-\alpha)\cdot n=3$.}
    \Description{A layered bipartite hard-instance gadget with $\alpha n=4$ and $(1-\alpha)n=3$. For each stage $k$, $A_k$ is connected triangularly to $B_k$ and completely to $C_k$, while $C_k$ is connected completely to $D_k$, $B_{k+1}$, and $C_{k+1}$.}
    \label{fig:upper}
\end{figure}

Next, we define the arrivals and deadlines of the vertices. Let there be $\ell$ stages. At the end of the $(k-1)$-th stage, all vertices in $\{A_i\}_{i \in [k-1]}$, $\{B_i\}_{i \in [k]}$, $\{C_i\}_{i \in [k]}$, $\{D_i\}_{i \in [k-1]}$ have arrived. Furthermore, all vertices in $\{A_i\}_{i \in [k-1]}$, $\{B_i\}_{i \in [k-1]}$, $\{C_i\}_{i \in [k-1]}$, $\{D_i\}_{i \in [k-1]}$ have left. In other words, the only remaining available vertices in the graph are those in $B_k \cup C_k$. (Refer to Figure~\ref{fig:construction-1}.)

We shall define the arrivals and deadlines in the $k$-th stage.
\begin{enumerate}
    \item (Fig.~\ref{fig:construction-2}) The vertices of $A_{k}$ arrive and leave immediately in the order of $a_{k1}, a_{k2}, \ldots$. When $a_{ki}$ arrives and departs, its neighbors $\{b_{kj}\}_{j\ge i} \cup C_k$ are still active (i.e., they have arrived but have not departed).
    \item (Fig.~\ref{fig:construction-2}) The vertices of $B_k$ leave. No neighbor is active for vertices in $B_k$.
    \item (Fig.~\ref{fig:construction-3}) The vertices of $D_k \cup B_{k+1} \cup C_{k+1}$ arrive and the vertices of $C_k$ leave one by one. The neighbors of $c_{ki}$ in $D_k \cup B_{k+1} \cup C_{k+1}$ are still active.
    \item (Fig.~\ref{fig:construction-3}) The vertices of $D_k$ leave. No neighbor is active for vertices in $D_k$.
\end{enumerate}
After the $\ell$-th stage, the vertices in the terminal buffer $B_{\ell+1}\cup C_{\ell+1}$ leave. At this point, all their neighbors in $C_\ell$ have already left, so no additional matching is made.

%Then, we give the arriving and leaving order for these defined vertices. We remark that the types of vertices are determined adversarially over time which are related to algorithms' decision. The instance is defined by a repeated process, which is start with the arrival of the first vertex of $A_k$. We need to assume there are $n_{cd} + n_{ab}$ vertices arrived before the process (i.e., they arrive in the previous process, or arrive at the beginning if we are in the first process.), which are the candidates of $B_k$ and $C_k$. Therefore, at the end of this process, we should prepare $n_{cd} + n_{ab}$ vertices that have not reached deadline for the next process. Next, we define what happens in the process between the arrival of $A_k$ and the preparation of the $n_{cd} + n_{ab}$ candidates, also see Figure~\ref{fig:construction} as an example.
\begin{figure}[ht]
    \center
    \begin{subfigure}{.15\textwidth}
        \includegraphics[scale = 0.4]{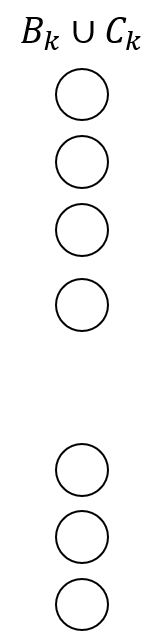}
        \caption{}
        \label{fig:construction-1}
    \end{subfigure}
    \begin{subfigure}{.15\textwidth}
        \includegraphics[scale = 0.4]{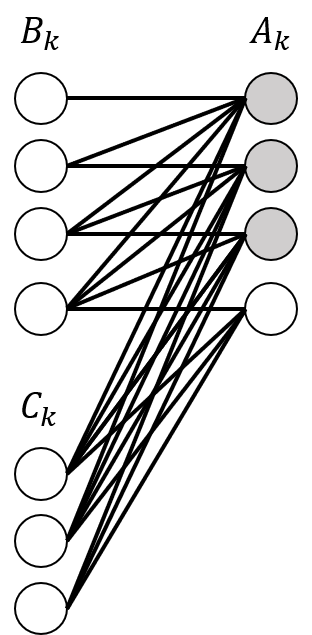}
         \caption{}
        \label{fig:construction-2}
    \end{subfigure}
    \begin{subfigure}{.3\textwidth}
        \includegraphics[scale = 0.4]{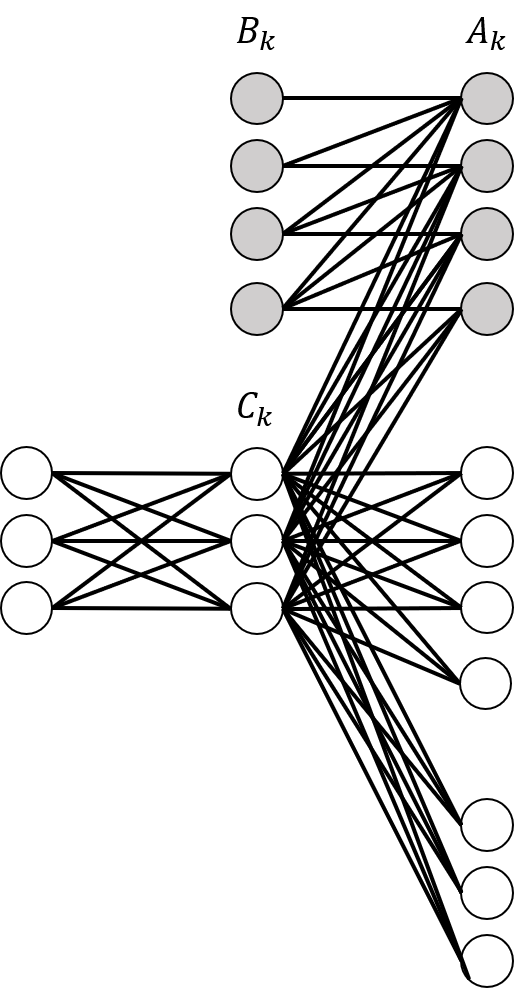}
        \caption{}
        \label{fig:construction-3}
    \end{subfigure}
    \begin{subfigure}{.3\textwidth}
        \includegraphics[scale = 0.4]{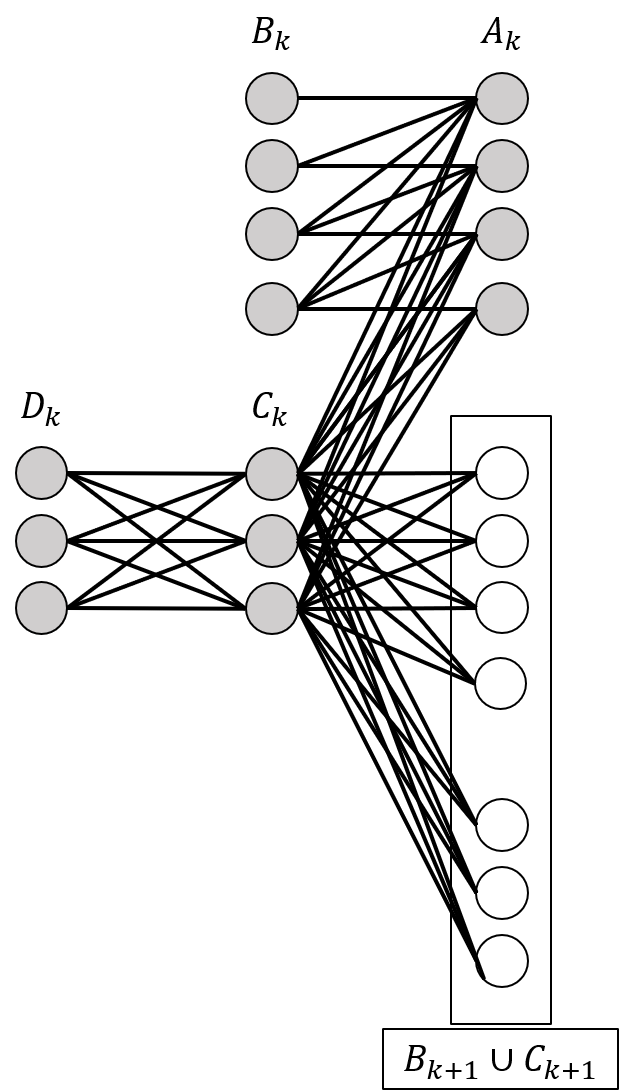}
        \caption{}
        \label{fig:construction-4}
    \end{subfigure}
    \caption{An illustration when $\alpha \cdot n =4$ and $(1-\alpha)\cdot n=3$: when the vertices turn grey, it indicates that they have already left.}
    \Description{{An illustration when $\alpha \cdot n =4$ and $(1-\alpha)\cdot n=3$: when the vertices turn grey, it indicates that they have already left.}}
    \label{fig:construction}
\end{figure}

Next, we analyze the competitive ratio.
{
% \color{red}
First, we observe that the maximum matching of the graph has size $n\ell$. Indeed, we can obtain a matching of size $n\ell$ by matching $A_k$ with $B_k$ and matching $C_k$ with $D_k$ for all $k \in [\ell]$. Conversely, $\bigcup_{k\in[\ell]}(A_k\cup C_k)$ is a vertex cover of size $n\ell$, including for the edges incident to the terminal buffer. Therefore, the maximum matching has size exactly $n\ell$.
}
% First, we observe that the maximum matching of the graph has size $n\ell$. Indeed, the graph has a perfect matching: we can match $A_k$ with $B_k$ and match $C_k$ with $D_k$ for all $k \in [\ell]$.
% 
Next, we assume that the algorithm only makes matching decisions on the departures of vertices, and that it fully matches a vertex unless there is no feasible neighbor. This assumption is without loss of generality by a simple exchange argument (e.g., see \cite{jacm/HuangKTWZZ20}).
A crucial feature of our construction is that, when a vertex reaches its deadline, its currently available neighbors are indistinguishable with respect to the unrevealed future. To formalize this symmetry, before the process begins, we independently apply a uniformly random permutation to the latent identities—or, equivalently, the future roles—of the vertices within each relevant group, revealing these roles only as they become observable through arriving edges. Consequently, conditional on any revealed history, the remaining neighbors are exchangeable. Following the symmetrization argument of~\cite{orl/EcklKLS21} and applying Yao's principle, it therefore suffices, for the purpose of analyzing expected performance, to consider the symmetric policy that distributes the departing vertex's remaining mass equally among its available neighbors; this is precisely the Water-Filling algorithm.

Based on the discussion above, we are left to calculate the performance of \wtf on the constructed instance. Since we treat all vertices in $B_k \cup C_k$ equally at the $(k-1)$-th stage, all vertices in $B_k \cup C_k$ have the same water level right before the $k$-th stage. Let $x_{k}$ be this water level. Initially, we have $x_1 = 0$. We derive the recurrence for $x_{k+1}$ in the following.
At the deadline of vertex $a_{ki}\in A_k$, it has $|\{b_{kj}\}_{j \ge i} \cup C_k| = n-i+1$ neighbors. Therefore, the increase in the water level of vertices of $C_k$ is equal to $\frac{1}{n-i+1}$ at the departure of $a_{ki}$. After the departure of the last vertex of $A_k$, the water level of vertices in $C_k$ becomes
\[
y_k = x_{k} + \frac{1}{n} + \frac{1}{n-1} + \cdots + \frac{1}{(1-\alpha)n + 1}~. 
\]
Next, at the departure of $C_k$, each vertex matches its remaining portion to all vertices in $D_k \cup B_{k+1} \cup C_{k+1}$ equally. Thus, we have
\[
x_{k+1} = (1 - y_k) \cdot \frac{1-\alpha}{2-\alpha} = \left(1 - \left(x_k + \frac{1}{n} + \frac{1}{n-1} + \cdots + \frac{1}{(1-\alpha)n+1}\right)\right) \cdot \frac{1-\alpha}{2-\alpha}~.
\]

Fix $n$ and let the number of stages $\ell$ go to infinity. The sequence $\{x_k\}$ converges to the solution $x^*$ of the following equation:
\[
x = \left(1 - \left(x + \frac{1}{n} + \frac{1}{n-1} + \cdots + \frac{1}{(1-\alpha)n+1}\right)\right) \cdot \frac{1-\alpha}{2-\alpha}~.
\]

Moreover, the total matching made in stage $k$ equals
\[
1 \cdot \alpha n + (1-y_k) \cdot (1-\alpha) n = \left(\alpha + (2-\alpha)\cdot x_{k+1}\right) \cdot n,
\]
where the term $1 \cdot \alpha n$ corresponds to the matchings made on the departures of $A_k$ and the term $(1-y_k)\cdot(1-\alpha) n$ corresponds to the matchings made on the departures of $C_k$.
% {
% \color{red}
% Hence, when $\ell \to \infty$, the ratio between the matching produced by \wtf and the optimal matching is
% \[
% \begin{aligned}
% \lim_{\ell\to\infty}
% \frac{1}{\ell}\sum_{k=1}^{\ell}
% \frac{\left(\alpha+(2-\alpha)\cdot x_{k+1}\right)n}{n}
% &=
% \lim_{\ell\to\infty}
% \frac{1}{\ell}\sum_{k=1}^{\ell}
% \left(\alpha+(2-\alpha)\cdot x_{k+1}\right)\\
% &=\alpha+(2-\alpha)\cdot x^*,
% \end{aligned}
% \]
% where the last equality follows from the convergence of $x_k$ to $x^*$.
% }
% % Hence, when $\ell \to \infty$, the ratio between the matching produced by \wtf and the optimal matching is
% % \[
% % \lim_{k \to \infty} \frac{\left(\alpha + (2-\alpha)\cdot x_{k+1}\right) \cdot n}{n} = \alpha + (2-\alpha)\cdot x^*~.
% % \]

% Finally, letting $n$ go to infinity, $x^*$ satisfies
% \[
% x^* = \left( 1-x^*+\ln(1-\alpha) \right) \cdot \frac{1-\alpha}{2-\alpha} \iff x^*= \frac{1-\alpha}{3-2\alpha} \cdot (1+\ln(1-\alpha)) ~,
% \]
% where we use the fact that
% \[
% \lim_{n\to \infty}\left( \frac{1}{n} + \frac{1}{n-1} + \cdots + \frac{1}{(1-\alpha)n+1} \right) = -\ln(1-\alpha)~.
% \]

% By optimizing $\alpha$, we conclude an upper bound of
% \[
% \min_{\alpha} \left\{ \alpha+(2-\alpha)\cdot x^* \right\} = \min_{\alpha} \left\{\alpha+\frac{(2-\alpha)(1-\alpha)}{3-2\alpha} \cdot (1+\ln(1-\alpha)) \right\} \approx 0.6132,
% \]
% when $\alpha \approx 0.43$.
{
% \color{red}
Hence, when $\ell \to \infty$, the ratio between the matching produced by \wtf and the optimal matching is
\[
\begin{aligned}
\lim_{\ell\to\infty}
\frac{1}{\ell}\sum_{k=1}^{\ell}
\frac{\left(\alpha+(2-\alpha)\cdot x_{k+1}\right)n}{n}
&=
\lim_{\ell\to\infty}
\frac{1}{\ell}\sum_{k=1}^{\ell}
\left(\alpha+(2-\alpha)\cdot x_{k+1}\right)\\
&=\alpha+(2-\alpha)\cdot x^*,
\end{aligned}
\]
where the last equality follows from the convergence of $x_k$ to $x^*$.

Finally, letting $n$ go to infinity, $x^*$ satisfies
\[
x^*
=
\left(1-x^*+\ln(1-\alpha)\right)
\cdot\frac{1-\alpha}{2-\alpha}
\iff
x^*
=
\frac{1-\alpha}{3-2\alpha}
\cdot(1+\ln(1-\alpha))~,
\]
where we use the fact that
\[
\lim_{n\to\infty}
\left(
\frac{1}{n}+\frac{1}{n-1}
+\cdots+\frac{1}{(1-\alpha)n+1}
\right)
=
-\ln(1-\alpha)~.
\]

Taking $\alpha=0.43$, we obtain
\[
\alpha+(2-\alpha)\cdot x^*
=
\alpha+
\frac{(2-\alpha)(1-\alpha)}{3-2\alpha}
\cdot(1+\ln(1-\alpha))
\approx 0.6131121
<0.6132.
\]
Since the above limits are approached by finite instances, the same strict inequality holds for sufficiently large finite $n$ (chosen as a multiple of $100$) and subsequently sufficiently large finite $\ell$. Therefore, no algorithm can achieve a competitive ratio strictly better than $0.6132$.
}

% \begin{suggestedrevision}[S1-38: Proposed conclusion to insert before the acknowledgments]
% \section{Conclusion}
% We established a $0.599$-competitive fractional algorithm for fully online
% matching on general graphs and a $0.6132$ upper bound that already holds on
% bipartite graphs.  Thus, the optimal competitive ratio lies in
% $[0.599,0.6132]$.  Our results also separate standard Water-Filling from
% algorithms that exploit inter-available edges and matching history.

% The main open problem is to close the remaining gap.  It would also be
% interesting to determine which parts of the history-based fractional policy
% can be rounded online without losing its advantage in the randomized integral
% setting.
% \end{suggestedrevision}
% \input{41eagerWTF}
% \input{42twodimfuction}
%%
%% The acknowledgments section is defined using the "acks" environment
%% (and NOT an unnumbered section). This ensures the proper
%% identification of the section in the article metadata, and the
%% consistent spelling of the heading.
% \begin{acks}
% Xiaowei Wu is funded by the Science and Technology Development Fund (FDCT), Macau SAR (file no. 0147/2024/RIA2, 001/2024/SKL, and 0002/2025/EQP), and the University of Macau (file no. MYRG-GRG2025-00033-IOTSC).
% \end{acks}

%%
%% The next two lines define the bibliography style to be used, and
%% the bibliography file.
\bibliographystyle{ACM-Reference-Format}
\bibliography{matching}

% \begin{suggestedrevision}[S1-39: Bibliography metadata cleanup]
% Replace abbreviated or incomplete entries by publisher or DBLP records before
% submission.  In particular, add the DOI and full venue metadata for the JACM
% 2020 and FOCS 2020 papers, and identify the 2026 Derakhshan--Yu work explicitly
% as an arXiv preprint unless a proceedings version is available.  Rerun BibTeX
% and clear all avoidable metadata warnings.
% \end{suggestedrevision}

%%
%% If your work has an appendix, this is the place to put it.
\appendix
\section{Factor-Revealing LP}
\subsection{Proof of Lemma~\ref{lemma:lower-bound-ratio-ewf}}
\label{sec:factor_revealing_lp_ewf}
Recall that we need a continuous and strictly increasing function $h$
satisfying the following inequalities.
\begin{align}
\forall q\in[0,1], \qquad & q \cdot h(q) - \int_0^q h(y) \dd{y} + 1 - q \geq \Gamma, \label{eqn:proof_v_earlier}\\
\forall q_u,q_v\in[0,1],\qquad & q_u\cdot h(q_u) - \int_0^{q_u}h(y) \dd y
+ q_v\cdot h(q_v) - \int_0^{q_v}h(y) \dd y \notag \\
& \qquad + \int_0^{1-q_u} h(y) \dd y + (1-h(q_u)) \cdot (1-q_v) \geq \Gamma=0.592, \label{eqn:proof_u_earlier}\\
& h(0) = 0, h(1) = 1.
\end{align}

% For any positive integer $n$, let $[0, 1]_n$ denote the set of multiples of $\frac{1}{n}$ between $0$ and $1$:
% \begin{equation*}
% $[0, 1]_n = \bigg\{ \frac{i}{n} : 0 \le i \le n \bigg\}.
% \end{equation*}
Fix $\varepsilon=1/10000$ before solving the LP. We seek grid values
$\{h(y)\}_{y\in[0,1]_n}$
satisfying $h(0)=0$, $h(1)=1$, and
\[
\frac{\varepsilon}{n}
\leq h\left(y+\frac1n\right)-h(y)
\leq\frac2n
\qquad
\text{for every }y\in[0,1]_n,\ y<1.
\]
Once these grid values are determined, we extend $h$ to the entire
interval $[0,1]$ by linear interpolation. Specifically, for each
$y\in[0,1)$, write
\[
y=\bar y+\frac{z_y}{n},
\qquad
\bar y\in[0,1]_n,\quad z_y\in[0,1),
\]
and define
\[
h(y)=(1-z_y)h(\bar y)
     +z_yh\left(\bar y+\frac1n\right).
\]
At the endpoint, we simply set $h(1)=1$. Thus, $h$ is obtained by
linearly interpolating between consecutive grid values.

For any $t\in[0,1]_n$, this interpolation gives
\begin{equation*}
\int_0^t h(y)\dd y
=
\sum_{\substack{y\in[0,1]_n\\y<t}}
\frac{h(y)+h(y+\frac1n)}{2n}.
\end{equation*}

It remains to determine the grid values
$\{h(y)\}_{y\in[0,1]_n}$. We do so using the following
factor-revealing LP. Its optimal solution defines, through the above
linear interpolation, a function $h$ satisfying the constraints in
Lemma~\ref{lemma:lower-bound-ratio-ewf}.
For $q\in[0,1]_n$, write
\[
q^+=\min\left\{q+\frac1n,1\right\}.
\]
For this fixed value of $\varepsilon$, use the finite LP
\begin{align}
{(LP_n)}\qquad\text{maximize}\qquad & r \nonumber \\
\text{subject to}\qquad
& \textstyle r \leq q \cdot h(q) - \int_0^q h(y)\dd y
+ 1-q - \frac{1}{4 n^2},
\qquad \forall q\in [0,1]_n
\label{constraint:v-arrive-first} \\
& \textstyle r \leq q \cdot h(q) - \int_0^q h(y) \dd y
+ p \cdot h(p) - \int_0^p h(y) \dd y
+ \int_0^{1-q} h(y) \dd y \nonumber\\
& \qquad \textstyle
+ \left(1-h(q^+)\right)\cdot (1-p) -\frac{3}{4n^2},
\qquad \forall q,p\in [0,1]_n
\label{constraint:u-arrive-first} \\
& \textstyle h(0)=0,\quad h(1)=1, \nonumber\\
& \textstyle
\frac{\varepsilon}{n}
\leq h\left(y+\frac1n\right)-h(y)
\leq\frac2n,
\qquad \forall y\in[0,1]_n,\quad y<1.
\tag{Monotonicity and Lipschitzness}
\end{align}

%Note that constraints~\eqref{constraint:v-arrive-first} are equivalent to
%\begin{equation*}
%r \leq q\cdot h(q) - \int_0^q h(y)\dd y - \frac{1}{4n^2}, \qquad \forall q\in[0,1]_n,
%\end{equation*}
%and constraints~\eqref{constraint:u-arrive-first} are equivalent to: $\forall q,p\in[0,1]_n$,
%\begin{equation*}
%r \leq q\cdot h(q) - \int_0^q h(y)\dd y + p\cdot h(p) - \int_0^p h(y)\dd y
% + \int_0^{1-q} h(y) \dd y + (1-h(q+\frac{1}{n})) \cdot (1-p) - \frac{3}{4n^2}.
%\end{equation*}

The following claim is verified using the Gurobi LP solver. 
% \footnote{Our code is available at \url{https://github.com/denil1111/Fully-Online-Maching-Improved-Algorithms}.}

\begin{lemma}\label{claim:lpn-ratio-ewf}
	For $n=2000$ and $\varepsilon=1/10000$, the revised finite LP has a
feasible solution with $r=0.592$.
\footnote{The factor-revealing LP implementations are available at
\url{https://github.com/denil1111/fully-online-matching-factor-revealing-lps}.}
\end{lemma}

% \begin{suggestedrevision}[Suggested replacement after rerunning the revised LP]
% For $n=1000$ and $\varepsilon=1/10000$, the revised finite LP has a feasible
% solution with $r=0.592$.  The LP generator, solver parameters,
% resulting solution, and an independent verification of all constraint
% residuals should be provided in the supplementary material.
% \end{suggestedrevision}

It remains to prove that an optimal solution of ${LP_n}$, whose
decision variables are $r$ and
$\{h(y)\}_{y\in[0,1]_n}$, defines the desired function $h$.
The increment constraints and linear interpolation make $h$ strictly
increasing, while the endpoint constraints give $h(0)=0$ and $h(1)=1$.
It therefore remains only to verify
constraints~\eqref{eqn:proof_v_earlier}
and~\eqref{eqn:proof_u_earlier}.

Let $h$ be defined by the optimal solution $\{ h(y) \}_{y\in[0,1]_n}$ of ${LP_n}$ with $n=2000$.
Fix any $q \in [0,1)$.
Let $q = \bar{q}+\frac{z_q}{n}$, where $\bar{q}\in[0,1]_n$ and $z_q\in [0,1)$.
Let $\hat{q} = \bar{q}+\frac{1}{n}$.
Observe that we have $q = (1-z_q)\cdot \bar{q}+z_q\cdot \hat{q}$ and $h(q) = (1-z_q)\cdot h(\bar{q}) + z_q\cdot h(\hat{q})$. At $q=1$, set
$\bar q=\hat q=1$ and $z_q=0$.  This avoids evaluating $h$ outside
$[0,1]$.

\begin{lemma}\label{claim:error-of-int-ewf}
	For $H(q)=qh(q)-\int_0^q h(y)\dd y$, we have
    \[
    H(q)-\bigl((1-z_q)H(\bar q)+z_qH(\hat q)\bigr)
    \in\left[-\frac{1}{4n^2},0\right].
    \]
\end{lemma}
\begin{proof}
	Recall that $h$ is defined by $\{h(y)\}_{y\in[0,1]_n}$, which satisfies the constraints of ${LP_n}$.

	We first consider the first term $q \cdot h(q)$. Observe that
	\begin{align*}
	& q\cdot h(q) - \Big( (1-z_q)\cdot\bar{q}\cdot h(\bar{q}) + z_q\cdot\hat{q}\cdot h(\hat{q}) \Big) = \frac{z_q(1-z_q)}{n}\cdot \Big( h(\bar{q})-h(\hat{q}) \Big) \in \left[-\frac{1}{2n^2},0\right],
	\end{align*}
	where the last step follows from the monotonicity and Lipschitzness of $\{h(y)\}_{y\in[0,1]_n}$.

	Next, we consider the second term $\int_0^q h(y) \dd{y}$. Observe that
	\begin{align*}
	& \quad \int_0^q h(y)\dd y - \left( (1-z_q)\int_0^{\bar{q}}h(y)\dd y + z_q\int_0^{\hat{q}}h(y)\dd y \right)
	= \int_{\bar{q}}^q h(y)\dd y - z_q\int_{\bar{q}}^{\hat{q}}h(y)\dd y \\
	& = \frac{z_q(1-z_q)}{2n}\cdot \Big( h(\bar{q}) - h(\hat{q}) \Big) \in \left[-\frac{1}{4n^2},0\right].
	\end{align*}

	Subtracting the second identity from the first one, we obtain
	\[
	H(q)-\bigl((1-z_q)H(\bar q)+z_qH(\hat q)\bigr)
	=\frac{z_q(1-z_q)}{2n}
	\bigl(h(\bar q)-h(\hat q)\bigr)
	\in\left[-\frac{1}{4n^2},0\right],
	\]
	which concludes the proof.
\end{proof}

\begin{lemma}
    \Cref{eqn:proof_v_earlier} is feasible with respect to the constructed function $h$.
\end{lemma}
% \paragraph{Feasibility of \Cref{eqn:proof_v_earlier}}
%	Note that the constraint is satisfied for all $q\in[0,1]_n$.
\begin{proof}
    By \Cref{constraint:v-arrive-first} and \Cref{claim:lpn-ratio-ewf}, we have
    \begin{align*}
    H(q) + 1-q & \geq  (1-z_q)\cdot \Big(H(\bar{q})+1-\bar{q}-\frac{1}{4n^2} \Big) +
    z_q\cdot \Big(H(\hat{q})+1-\hat{q}-\frac{1}{4n^2} \Big) \\
    & \geq (1-z_q)\cdot \Gamma + z_q\cdot \Gamma = \Gamma.
    \end{align*}
\end{proof}

\begin{lemma}
    \Cref{eqn:proof_u_earlier} is feasible with respect to the constructed function $h$.
\end{lemma}

\begin{proof}
    Let $\bar{p},\hat{p}, z_p$ (for $p$) be defined similarly as $\bar{q},\hat{q}, z_q$ (for $q$). For each grid point $s\in[0,1]_n$, let $s^+=\min\{s+1/n,1\}$. Note that
    \[
        1-p=(1-z_p)(1-\bar{p})+z_p(1-\hat{p}).
    \]
    By \Cref{claim:error-of-int-ewf} and the monotonicity of $h$, we have
    \begin{align*}
     & \quad H(q) + H(p) + \int_0^{1-q} h(y) \dd y + (1-h(q)) \cdot (1-p) \\
    & \geq H(q) + H(p) + \int_0^{1-q} h(y) \dd y + (1-h(\hat{q})) \cdot (1-p) \\
    & = H(q) + H(p) + \int_0^{1-q} h(y) \dd y
    + (1-z_p)(1-h(\hat{q}))(1-\bar{p}) + z_p(1-h(\hat{q}))(1-\hat{p}) \\
    & \geq \Big((1-z_q)H(\bar{q})+z_qH(\hat{q})-\frac{1}{4n^2}\Big)
    + \Big((1-z_p)H(\bar{p})+z_pH(\hat{p})-\frac{1}{4n^2}\Big) \\
    & \quad + \Big((1-z_q)\int_0^{1-\bar{q}} h(y) \dd y + z_q\int_0^{1-\hat{q}} h(y) \dd y
    + \frac{z_q(1-z_q)}{2n}\big(h(1-\hat{q})-h(1-\bar{q})\big)\Big) \\
    & \quad + (1-z_p)(1-h(\hat{q}))(1-\bar{p}) + z_p(1-h(\hat{q}))(1-\hat{p}) \\
    & \geq \Big((1-z_q)H(\bar{q})+z_qH(\hat{q})-\frac{1}{4n^2}\Big)
    + \Big((1-z_p)H(\bar{p})+z_pH(\hat{p})-\frac{1}{4n^2}\Big) \\
    & \quad + \Big((1-z_q)\int_0^{1-\bar{q}} h(y) \dd y + z_q\int_0^{1-\hat{q}} h(y) \dd y-\frac{1}{4n^2}\Big) \\
    & \quad + (1-z_q)(1-z_p)\Big(1-h\big(\bar{q}^{+}\big)\Big)(1-\bar{p})
    + (1-z_q)z_p\Big(1-h\big(\bar{q}^{+}\big)\Big)(1-\hat{p}) \\
    & \quad + z_q(1-z_p)\Big(1-h\big(\hat{q}^{+}\big)\Big)(1-\bar{p})
    + z_qz_p\Big(1-h\big(\hat{q}^{+}\big)\Big)(1-\hat{p}) \\
    & = (1-z_q)(1-z_p)\Big(H(\bar{q})+H(\bar{p})+\int_0^{1-\bar{q}}h(y)\dd y
    + \Big(1-h\big(\bar{q}^{+}\big)\Big)(1-\bar{p})-\frac{3}{4n^2}\Big) \\
    & \quad + (1-z_q)z_p\Big(H(\bar{q})+H(\hat{p})+\int_0^{1-\bar{q}}h(y)\dd y
    + \Big(1-h\big(\bar{q}^{+}\big)\Big)(1-\hat{p})-\frac{3}{4n^2}\Big) \\
    & \quad + z_q(1-z_p)\Big(H(\hat{q})+H(\bar{p})+\int_0^{1-\hat{q}}h(y)\dd y
    + \Big(1-h\big(\hat{q}^{+}\big)\Big)(1-\bar{p})-\frac{3}{4n^2}\Big) \\
    & \quad + z_qz_p\Big(H(\hat{q})+H(\hat{p})+\int_0^{1-\hat{q}}h(y)\dd y
    + \Big(1-h\big(\hat{q}^{+}\big)\Big)(1-\hat{p})-\frac{3}{4n^2}\Big) \\
    & \geq (1-z_q)(1-z_p)\Gamma + (1-z_q)z_p\Gamma
    + z_q(1-z_p)\Gamma + z_qz_p\Gamma = \Gamma,
    \end{align*}
    where the last inequality follows from \Cref{constraint:u-arrive-first}, applied to the four grid-point pairs
    $(\bar{q},\bar{p})$, $(\bar{q},\hat{p})$, $(\hat{q},\bar{p})$, and $(\hat{q},\hat{p})$.
\end{proof}

\subsection{Proof of Lemma~\ref{lem:opth}}
\label{app:factor_lp}
In this section, we construct a continuous function $h$ whose rows are strictly
increasing, whose diagonal is a strictly increasing bijection, and
which satisfies the two factor-revealing inequalities for
$\Gamma=0.599$. Let
\begin{align*}
    &\Phi_1(\tau_v,\theta_v) = \theta_v \cdot h(\tau_v, \theta_v) - \int_{\tau_v}^{\theta_v} h(\tau_v, y_v) dy_v + 1-\theta_v; \\
    &\Phi_2(\tau_u,\theta_u,\tau_v,\theta_v) = \left(
    \begin{aligned}
        &\theta_u \cdot h(\tau_u, \theta_u) - \int_{\tau_u}^{\theta_u} h(\tau_u, y_u) dy_u + \theta_v \cdot h(\tau_v, \theta_v)\\
        & - \int_{\tau_v}^{\theta_v} h(\tau_v, y_v) dy_v + (1-h(\tau_u, \theta_u)) \cdot (1-\theta_v)
    \end{aligned}
    \right).
\end{align*}
The two constraints are defined as follows:
\begin{align}
	& \forall 0\le \tau_v \le \theta_v \le 1: &\Phi_1 \ge \Gamma; \tag{\ref{eqn:h-feasible-1}} \\
    & \forall 0\le \tau_u \le \theta_u \le 1, 1-\theta_u\le \tau_v \le \theta_v \le 1 :\quad  & \Phi_2 \ge \Gamma; \tag{\ref{eqn:h-feasible-2}}
\end{align}

\paragraph{The Factor-Revealing LP and the Construction of $h$} Fix an integer $n$. Let $1/n$ be the step size, and let
$[0,1]_n=\{i/n:i=0,\ldots,n\}$. The LP variables are the values
$h(\tau,\theta)$ for $\tau,\theta\in[0,1]_n$. The remaining values of
$h$ are defined by bilinear interpolation. In particular, for
$x\in[0,1)$, let $\bar x$ and $\hat x=\bar x+1/n$ be the two adjacent
grid points and let $z(x)=n(x-\bar x)$. At $x=1$, set
$\bar x=\hat x=1$ and $z(x)=0$. Thus,
$x=(1-z(x))\bar x+z(x)\hat x$. We define
\[
\begin{aligned}
h(\tau,\theta)={}&
(1-z(\tau))(1-z(\theta))h(\bar\tau,\bar\theta)
+(1-z(\tau))z(\theta)h(\bar\tau,\hat\theta)\\
&+z(\tau)(1-z(\theta))h(\hat\tau,\bar\theta)
+z(\tau)z(\theta)h(\hat\tau,\hat\theta).
\end{aligned}
\]

For grid points $\tau,\theta\in[0,1]_n$, we use the signed
trapezoidal convention
\[
\int_\tau^\theta h(\tau,y)\dd y=
\begin{cases}
\displaystyle
\sum_{\substack{y\in[0,1]_n\\ \tau\leq y<\theta}}
\frac{h(\tau,y)+h(\tau,y+\frac1n)}{2n},
&\tau<\theta,\\[2ex]
0,&\tau=\theta,\\[1ex]
\displaystyle
-\sum_{\substack{y\in[0,1]_n\\ \theta\leq y<\tau}}
\frac{h(\tau,y)+h(\tau,y+\frac1n)}{2n},
&\tau>\theta.
\end{cases}
\]
The last case is needed only for the one-cell halo used below.
With this convention, $\Phi_1$ and $\Phi_2$ remain linear functions
of the LP variables at all grid points used in the constraints.

Fix $\varepsilon=1/10000$ before solving the LP.
We use the following finite LP:
\begin{align}
\text{maximize}\qquad & \Gamma \nonumber\\
\text{subject to}\qquad
& \Gamma\leq\Phi_1(\tau_v,\theta_v)-\frac{2}{n^2},
&& \forall \tau_v,\theta_v\in[0,1]_n:\
\tau_v\leq\theta_v+\frac1n,
\label{constraint:v-arrive-first-2d}\\
& \Gamma\leq
\Phi_2(\tau_u,\theta_u,\tau_v,\theta_v)-\frac4{n^2},
&& \begin{aligned}
&\forall \tau_u,\theta_u,\tau_v,\theta_v\in[0,1]_n:\\
&\tau_u\leq\theta_u+\frac1n,\quad
\tau_v\leq\theta_v+\frac1n,\\
&\tau_v\geq1-\theta_u-\frac1n,
\end{aligned}
\label{constraint:u-arrive-first-2d}\\
& h(\tau,0)=0,\qquad h(\tau,1)=1,
&& \forall \tau\in[0,1]_n,\nonumber\\
& \frac{\varepsilon}{n}
\leq h\left(\tau,y+\frac1n\right)-h(\tau,y)
\leq\frac4n,
&& \forall \tau,y\in[0,1]_n:\ y<1,
\tag{Monotonicity and row Lipschitzness}\\
& h\left(y+\frac1n,\theta\right)-h(y,\theta)
\leq\frac4n,
&& \forall y,\theta\in[0,1]_n:\ y<1,
\tag{First-coordinate bound}\\
& h\left(y,y+\frac1n\right)
+h\left(y+\frac1n,y\right)-2h(y,y)
\geq\frac{\varepsilon}{n},
&& \forall y\in[0,1]_n:\ y<1,
\tag{Diagonal monotonicity I}\\
& 2h\left(y+\frac1n,y+\frac1n\right)
-h\left(y,y+\frac1n\right)
-h\left(y+\frac1n,y\right)
\geq\frac{\varepsilon}{n},
&& \forall y\in[0,1]_n:\ y<1.
\tag{Diagonal monotonicity II}
\end{align}

The relaxed index conditions form a one-cell halo around the continuous
feasible region. If $\tau\leq\theta$, every interpolation corner
satisfies $\tau'\leq\theta'+1/n$. Similarly, if
$1-\theta_u\leq\tau_v$, every corner satisfies
$\tau_v'\geq1-\theta_u'-1/n$. Hence every corner used in the
two- and four-dimensional interpolations appears in an LP constraint.

The last two constraints guarantee strict monotonicity of the diagonal.
Indeed, for $y\in[0,1]_n$ with $y<1$, let
\[
d_y(z)=h\left(y+\frac{z}{n},y+\frac{z}{n}\right),
\qquad 0\leq z\leq1.
\]
Bilinear interpolation gives
\begin{align*}
d_y'(z)={}&(1-z)\left(
h\left(y,y+\frac1n\right)
+h\left(y+\frac1n,y\right)-2h(y,y)\right)\\
&+z\left(
2h\left(y+\frac1n,y+\frac1n\right)
-h\left(y,y+\frac1n\right)
-h\left(y+\frac1n,y\right)\right)
\geq\frac{\varepsilon}{n}.
\end{align*}
Thus $\tau\mapsto h(\tau,\tau)$ is a strictly increasing bijection.
The row constraints and bilinear interpolation similarly imply that
$\theta\mapsto h(\tau,\theta)$ is strictly increasing for every fixed
$\tau$. The parameter $\varepsilon$ is fixed data and is not an LP
variable.

For $n=200$ and $\varepsilon=1/10000$, solving the revised finite LP
gives an objective value of $0.5997$. \footnote{The factor-revealing LP implementations are available at
\url{https://github.com/denil1111/fully-online-matching-factor-revealing-lps}.} In particular, the LP has a
feasible solution with $\Gamma=0.599$.
In the following, we prove that the optimal solution of the above LP
constructs a feasible function $h$, i.e., satisfies
Equations~\eqref{eqn:h-feasible-1} and~\eqref{eqn:h-feasible-2}, thereby
deriving the $0.599$ competitive ratio.

These constraints appear directly in the LP when the parameters (i.e.,
$\tau_u$, $\tau_v$, $\theta_u$, and $\theta_v$) are all inside
$[0,1]_n$; it remains to consider the case when some parameters are
outside $[0,1]_n$. For each interpolated coordinate $x$, define
\[
x^{(0)}=\bar x,\qquad x^{(1)}=\hat x,\qquad
\lambda_x(0)=1-z(x),\qquad\lambda_x(1)=z(x).
\]
Then write
\[
\widetilde\Phi_1(\tau_v,\theta_v)
=\sum_{a,b\in\{0,1\}}
\lambda_{\tau_v}(a)\lambda_{\theta_v}(b)
\Phi_1\bigl(\tau_v^{(a)},\theta_v^{(b)}\bigr)
\]
and
\begin{align*}
\widetilde\Phi_2(\tau_u,\theta_u,\tau_v,\theta_v)
=\sum_{a,b,c,d\in\{0,1\}}
&\lambda_{\tau_u}(a)\lambda_{\theta_u}(b)
\lambda_{\tau_v}(c)\lambda_{\theta_v}(d)\\
&\quad\cdot
\Phi_2\bigl(\tau_u^{(a)},\theta_u^{(b)},
\tau_v^{(c)},\theta_v^{(d)}\bigr).
\end{align*}
By the halo index conditions, every corner in the two sums above
belongs to the enlarged LP constraint set.  Since the coefficients are
nonnegative and sum to one, these are convex combinations, and hence
\[
\widetilde\Phi_1\geq\Gamma+\frac{2}{n^2},
\qquad
\widetilde\Phi_2\geq\Gamma+\frac4{n^2}.
\]
The key point is captured by the following two lemmas:
\begin{lemma}
\label{lem:phi1-bound}
    $\Phi_1(\tau_v,\theta_v) \geq \tilde{\Phi}_1(\tau_v,\theta_v) - \frac{2}{n^2} \geq \Gamma$.
\end{lemma}
\begin{lemma}
\label{lem:phi2-bound}
    $\Phi_2(\tau_u,\theta_u,\tau_v,\theta_v) \geq \tilde{\Phi}_2(\tau_u,\theta_u,\tau_v,\theta_v) - \frac{4}{n^2} \geq \Gamma$.
\end{lemma}

To prove these two lemmas, we define a crucial function below:
\begin{align*}
&H(\tau,\theta) = \theta h(\tau,\theta) - \int_{\tau}^{\theta} h(\tau,y) dy; \\
&\tilde{H}(\tau,\theta) = \left(
\begin{aligned}
    &(1-z(\tau))(1-z(\theta))H(\bar{\tau},\bar{\theta})
    + (1-z(\tau))\cdot z(\theta) H(\bar{\tau},\hat{\theta}) \\
    &+ z(\tau)\cdot(1-z(\theta))H(\hat{\tau},\bar{\theta})
    + z(\tau)\cdot z(\theta) H(\hat{\tau},\hat{\theta})
\end{aligned}
\right).
\end{align*}
We claim the following lemma.

\begin{lemma}
    $H(\tau,\theta) \geq \tilde{H}(\tau,\theta) - \frac{2}{n^2}$.
\end{lemma}
\begin{proof}
    We first restrict $\tau \in [0,1]_n$ but relax $\theta$ to be in $[0,1]$, and prove that
    $$
    H(\tau,\theta) \geq (1-z(\theta))\cdot H(\tau,\bar{\theta}) + z(\theta)\cdot H(\tau,\hat{\theta}) - \frac{1}{2n^2}.
    $$
    Recall that $H(\tau,\theta) = \theta h(\tau,\theta) - \int_{\tau}^{\theta} h(\tau,y) dy$. Considering the first term, we have
    \begin{align*}
        & \theta h(\tau,\theta) - (1-z(\theta))\cdot \bar{\theta} h(\tau,\bar{\theta}) - z(\theta)\cdot \hat{\theta} h(\tau,\hat{\theta}) \\
        ={}& \theta\cdot(1-z(\theta))\cdot h(\tau,\bar{\theta}) + \theta\cdot z(\theta)h(\tau,\hat{\theta}) - (1-z(\theta))\cdot \bar{\theta} h(\tau,\bar{\theta}) - z(\theta)\cdot \hat{\theta} h(\tau,\hat{\theta}) \\
        ={}& \frac{z(\theta)}{n} \cdot (1-z(\theta)) \cdot h(\tau,\bar{\theta}) - \frac{1-z(\theta)}{n} \cdot z(\theta) \cdot h(\tau,\hat{\theta}) \\
        ={}& \frac{z(\theta)(1-z(\theta))}{n} \cdot (h(\tau,\bar{\theta}) -h(\tau,\hat{\theta})).
    \end{align*}
    % \begin{align*}
    %     & - \int_{\tau}^{\theta} h(\tau,y)dy + (1-z(\theta)) \cdot \int_{\tau}^{\bar{\theta}} h(\tau,y)dy + z(\theta) \cdot \int_{\tau}^{\hat{\theta}} h(\tau,y)dy \\
    %     = & - \int_{\bar{\theta}}^{\theta} h(\tau,y)dy + z(\theta) \cdot \int_{\hat{\theta}}^{\bar{\theta}} h(\tau,y)dy \\
    %     \geq & 0. \qquad \text{(By the Monotonicity of $h$)}
    % \end{align*}
{
% new version
    % \color{red}
    For the other term, linear interpolation gives
    \begin{align*}
        & - \int_{\tau}^{\theta} h(\tau,y)\dd y
        + (1-z(\theta)) \int_{\tau}^{\bar{\theta}} h(\tau,y)\dd y
        + z(\theta) \int_{\tau}^{\hat{\theta}} h(\tau,y)\dd y \\
        ={}&
        - \int_{\bar{\theta}}^{\theta} h(\tau,y)\dd y
        + z(\theta) \int_{\bar{\theta}}^{\hat{\theta}} h(\tau,y)\dd y \\
        ={}&
        -\frac{z(\theta)(1-z(\theta))}{2n}
        \bigl(h(\tau,\bar\theta)-h(\tau,\hat\theta)\bigr).
    \end{align*}
    Subtracting the integral interpolation error from the first-term
    interpolation error, we obtain
    \[
    H(\tau,\theta)-(1-z(\theta))H(\tau,\bar\theta)
    -z(\theta)H(\tau,\hat\theta)
    =\frac{z(\theta)(1-z(\theta))}{2n}
    \bigl(h(\tau,\bar\theta)-h(\tau,\hat\theta)\bigr)
    \geq-\frac{1}{2n^2},
    \]
    where the last inequality follows from the row Lipschitz bound.
}
    This proves the first claim. Next, for any fixed
    $\theta\in[0,1]$, we allow $\tau$ to be in $[0,1]$ and prove that
    $$
    H(\tau,\theta) \geq (1-z(\tau))\cdot H(\bar{\tau},\theta) + z(\tau)\cdot H(\hat{\tau},\theta) - \frac{3}{2n^2}.
    $$
    Similarly, considering the first term, we have
    \begin{align*}
        & \theta h(\tau,\theta) - (1-z(\tau))\cdot \theta h(\bar{\tau},\theta) - z(\tau)\cdot \theta h(\hat{\tau},\theta) = 0.
    \end{align*}
    For the second term, we have
    \begin{align*}
        & - \int_{\tau}^{\theta} h(\tau,y)dy + (1-z(\tau)) \cdot \int_{\bar{\tau}}^{\theta} h(\bar{\tau},y)dy + z(\tau) \cdot \int_{\hat{\tau}}^{\theta} h(\hat{\tau},y)dy \\
        = & - (1-z(\tau)) \cdot \int_{\tau}^{\theta} h(\bar{\tau},y)dy - z(\tau) \cdot \int_{\tau}^{\theta} h(\hat{\tau},y)dy \\
        & + (1-z(\tau)) \cdot \int_{\bar{\tau}}^{\theta} h(\bar{\tau},y)dy + z(\tau) \cdot \int_{\hat{\tau}}^{\theta} h(\hat{\tau},y)dy \\
        ={}& (1-z(\tau)) \cdot \int_{\bar{\tau}}^{\tau} h(\bar{\tau},y)dy  - z(\tau) \cdot \int_{\tau}^{\hat{\tau}} h(\hat{\tau},y)dy \\
        ={}& (1-z(\tau)) \cdot \frac{z(\tau)}{2n} \cdot(h(\bar{\tau},\bar{\tau}) + h(\bar{\tau},\tau))  - z(\tau) \cdot \frac{1-z(\tau)}{2n} \cdot(h(\hat{\tau},\tau) + h(\hat{\tau},\hat{\tau}))\\
        \geq & -\frac{z(\tau)\cdot(1-z(\tau))}{2n}\cdot \frac{12}{n} \geq -\frac{3}{2n^2}  \qquad \text{(By the Lipschitzness of $h$)}
    \end{align*}

The first-coordinate bound, extended to arbitrary second coordinates by
linear interpolation, gives
\[
h(\hat\tau,\tau)-h(\bar\tau,\tau)\leq\frac4n \quad \text{and} \quad h(\hat\tau,\hat\tau)-h(\bar\tau,\bar\tau)
\leq h(\hat\tau,\hat\tau)-h(\bar\tau,\hat\tau) + h(\bar\tau,\hat\tau)-h(\bar\tau,\bar\tau)
\leq\frac8n.
\]
Thus the difference of the two sums inside the preceding trapezoidal
expression is at most $12/n$, which yields the stated
$3/(2n^2)$ error.  The lower margin
$\varepsilon/n$ is used only for strict monotonicity and
does not add another approximation loss.
Applying the $\tau$-interpolation bound first and then the
$\theta$-interpolation bound at $\bar\tau$ and $\hat\tau$, whose
weights sum to one, gives
$H(\tau,\theta)\geq\tilde H(\tau,\theta)-2/n^2$.
\end{proof}

Finally, we conclude the proof of Lemma~\ref{lem:phi1-bound} and Lemma~\ref{lem:phi2-bound}.
\begin{proof}[Proof of Lemma~\ref{lem:phi1-bound}]
\begin{align*}
    \Phi_1(\tau_v,\theta_v)
    = H(\tau_v,\theta_v) + 1 - \theta_v
    \geq \tilde{H}(\tau_v,\theta_v) + 1 - \theta_v - \frac{2}{n^2} \geq \tilde{\Phi}_1(\tau_v,\theta_v) - \frac{2}{n^2} \geq \Gamma.
\end{align*}
\end{proof}
\begin{proof}[Proof of Lemma~\ref{lem:phi2-bound}]
\begin{align*}
    &\theta_u \cdot h(\tau_u, \theta_u) - \int_{\tau_u}^{\theta_u} h(\tau_u, y_u) dy_u + \theta_v \cdot h(\tau_v, \theta_v) - \int_{\tau_v}^{\theta_v} h(\tau_v, y_v) dy_v + (1-h(\tau_u, \theta_u)) \cdot (1-\theta_v) \\
    ={}& H(\tau_u,\theta_u) + H(\tau_v,\theta_v) + (1-h(\tau_u,\theta_u))\cdot(1-\theta_v)\\
    \geq& \tilde{H}(\tau_u,\theta_u) + \tilde{H}(\tau_v,\theta_v) + (1-h(\tau_u,\theta_u))\cdot(1-\theta_v) - \frac{4}{n^2}\\
    ={}& \tilde{H}(\tau_u,\theta_u) + \tilde{H}(\tau_v,\theta_v) + 1-h(\tau_u,\theta_u) - \theta_v + \theta_v h(\tau_u,\theta_u) - \frac{4}{n^2} \\
    ={}& \tilde{\Phi}_2(\tau_u,\theta_u,\tau_v,\theta_v) - \frac{4}{n^2} \geq \Gamma.
\end{align*}
The cross term
\[
(1-h(\tau_u,\theta_u))(1-\theta_v)
\]
is multi-affine in the interpolation coordinates.  Its corner
interpolation is therefore exact, so it introduces no error beyond the
two $H$-terms.  Consequently, the total loss is $4/n^2$.
\end{proof}

\section{Primal-Dual Connection between the Upper and Lower Bounds} \label{sec:wtf_gh_dual}

We provide an interesting primal-dual connection between the primal-dual analysis in Section~\ref{sec:wtf_lb} and the hard instance in Section~\ref{sec:wtf_ub}, which guides our choice of the function $h$ in Section~\ref{sec:wtf_ub}. The following derivation is heuristic and only explains how we
discovered the hard instance; no main result relies on the functional
duality below.

Recall the following lower bound established in the proof of Lemma~\ref{thm:wtf_lower}:
\begin{align*}
\alpha_u + \alpha_v 
\ge \min \Big\{ 
& \int_0^1 g(x)\,dx, 
\min_{p_u,\,x_v} 
\bigl\{
    \int_0^{p_u} g(x)\,dx
    + (1-p_u)\bigl(1-g(x_v)\bigr)
    + \int_0^{x_v} g(x)\,dx
\bigr\}
\Big\}.
\end{align*}
For the candidate relevant to our construction, the first term in the
above minimum is nonbinding, so we formally focus on the second family:
\begin{align*}
\max_g \quad & r \\
\text{s.t.} \quad 
& r \le \int_0^x g(y)\,dy + \int_0^z g(y)\,dy + (1-x)\bigl(1-g(z)\bigr), 
\quad \forall\, x,z \in [0,1].
\end{align*}
The tight constraints suggest the following reduced program:
\begin{align*}
(P)\quad 
\max_g \quad & r \\
\text{s.t.} \quad 
& r \le \int_0^x g(y)\,dy + \int_0^{c-x} g(y)\,dy + (1-x)\bigl(1-g(c-x)\bigr),
\quad \forall\, x\in[0,c],
\end{align*}
where $c = 2 - \sqrt{2}$.
For the relevant solution, this reduced program yields the same
candidate value, and formal complementary slackness suggests that tight
pairs satisfy $p_u+x_v=c$ when \wtf is executed.

According to the instance structure and the argument in Section~\ref{sec:wtf_ub}, it suffices to find a function $h:[0,1]\to[0,1]$ satisfying
\[
\int_0^x \frac{1-\tau(y)}{1-y+h(y)}\,dy = \tau(h(x)) = c - \tau(x),
\quad \forall\, x\in[0,1].
\]
Here, $\tau:[0,1]\to[0,c]$ denotes the stationary water level, yielding the first equality.
Moreover, the perfect partner corresponding to $x$ has water level $\tau(h(x))$, which we require to equal $c - \tau(x)$, yielding the second equality.
Therefore, $h(x) = \tau^{-1}\bigl(c - \tau(x)\bigr)$. 
Let $f(x) = \tau^{-1}(x)$.
Differentiating the above relation, we require the existence of functions $f,\tau,h$ satisfying
\begin{align*}
\frac{1-\tau(x)}{1-x+h(x)} = -\tau'(x)
&\;\Leftrightarrow\;
\frac{1-\phi}{1-f(\phi)+f(c-\phi)} = -\frac{1}{f'(\phi)} \\
&\;\Leftrightarrow\;
1 - f(\phi) + f(c-\phi) + (1-\phi)f'(\phi) = 0,
\quad \forall\, \phi\in[0,c].
\end{align*}

% Now consider the dual program of $(P)$:
% \begin{align*}
% \min_q \quad & \int_0^c (1-x)q(x)\,dx \\
% \text{s.t.} \quad 
% & 1 - \int_0^c q(x)\,dx \le 0, \\
% & \int_0^x q(y)\,dy + \int_{c-x}^c q(y)\,dy - (1-x)q(x) \le 0,
% \quad \forall\, x\in[0,c].
% \end{align*}
% By the primal--dual theory, the optimal dual solution $q(x)$ satisfies
% \[
% \int_0^c q(x)\,dx = 1,
% \qquad
% \int_0^x q(y)\,dy + \int_{c-x}^c q(y)\,dy + (1-x)q(x) = 0.
% \]
% Let $Q(x) = 1 - \int_0^x q(y)\,dy = \int_x^c q(y)\,dy$.
% Then
% \[
% 1 - Q(x) + Q(c-x) + (1-x)Q'(x) = 0,
% \quad \forall\, x\in[0,c],
% \]
% which is exactly the same functional equation required for $f$.
{
% \color{red}
Now consider the dual program of $(P)$:
\begin{align*}
\min_{q\geq 0} \quad & \int_0^c (1-x)q(x)\,dx \\
\text{s.t.} \quad
& 1 - \int_0^c q(x)\,dx \le 0, \\
& \int_0^x q(y)\,dy
  + \int_{c-x}^c q(y)\,dy
  - (1-x)q(x) \le 0,
\quad \forall\, x\in[0,c].
\end{align*}
Formally invoking strong duality and complementary slackness makes the
relevant constraints tight at the optimal dual solution $q(x)$, giving
$
\int_0^c q(x)\,dx = 1$,
and 
$
\int_0^x q(y)\,dy
+\int_{c-x}^c q(y)\,dy
-(1-x)q(x)=0.
$
Let
$
Q(x)=1-\int_0^x q(y)\,dy=\int_x^c q(y)\,dy.
$
Since $q(x)=-Q'(x)$, the preceding equality becomes
\[
1-Q(x)+Q(c-x)+(1-x)Q'(x)=0,
\quad \forall\,x\in[0,c],
\]
which is exactly the same functional equation required for $f$.
}

\end{document}